\documentclass[11pt]{article}
\usepackage[margin=1in]{geometry}
\usepackage{booktabs}
\usepackage{array}
\usepackage{multirow}
\usepackage{multicol}
\usepackage{comment}
\usepackage{bbm}
\usepackage{float}
\usepackage{subcaption}

\usepackage{graphicx} 

\usepackage{multirow,amssymb,amsmath,tabularx, amsthm, amsfonts,color}
\usepackage{import}
\usepackage{placeins}
\usepackage[ruled,linesnumbered]{algorithm2e}
\usepackage{chngcntr, apptools}
\usepackage{mathtools}

\let\oldnl\nl
\newcommand{\nonl}{\renewcommand{\nl}{\let\nl\oldnl}}

\newlength\mylen

\newcommand{\I}[1]{\mathbb{I}\left[#1\right]}

\newtheorem{proposition}{Proposition}

\newcommand\numberthis{\addtocounter{equation}{1}\tag{\theequation}}

\usepackage{natbib}
\usepackage[hidelinks]{hyperref}
\usepackage{url}

\title{Regularized covariance estimation from partially observed interferometric data}
\date{}
\author{Teresa Bortolotti\textsuperscript{1,*}, Roberta Troilo\textsuperscript{1,*}, Francesco Casu\textsuperscript{2}, Simone Vantini\textsuperscript{1},\\ and Alessandra Menafoglio\textsuperscript{1}
}

\begin{document}

\maketitle
\let\thefootnote\relax
\footnotetext{\hspace{-0.57cm}\textsuperscript{1}MOX, Department of Mathematics, Politecnico di Milano, Piazza Leonardo da Vinci 32, 20133, Milan, Italy.\\
\textsuperscript{2}Istituto per il Rilevamento Elettromagnetico dell'Ambiente, National Research Council, Via A. Corti 12, 20133, Milan, Italy.\\
\textsuperscript{*}The authors contributed equally to this work.}


\begin{abstract}
The Small BAseline Subset technique provides remote measurements of ground displacement with high spatial resolution, making it a key tool for monitoring geophysical processes in hazard-prone areas. An effective analysis of this type of data requires reliable estimation of their second-order structure, which is difficult to achieve because the measurements are systematically missing over relatively large portions of the investigated areas.
We tackle the problem from a functional data analysis perspective and treat the observations as partially observed functional data with two-dimensional domain.
To properly characterize the data, we introduce the fragmented regime of partial observation, where parts of the curves are systematically missing across replicates. 
For this regime, we propose a novel method for covariance estimation, formulating the task as a matrix completion problem with Laplacian regularization. The estimator is nonparametric and free from stationarity or isotropy assumptions.
Extensive simulations show that our method achieves consistently low estimation error across a range of covariance structures. Application to ground displacement data relative to the Phlegraean Fields demonstrates its ability to recover meaningful spatial dependence patterns, highlighting its potential for environmental risk assessment and monitoring.\\

\noindent \textbf{Keywords:} covariance estimation; functional data analysis; matrix laplacian regularization; non-stationarity; partially observed data; remote sensing.
\end{abstract} 

\bigskip

\section{Introduction}
\label{sec:introduction}
\paragraph{Background and motivation}
Spaceborne radar interferometry has emerged as a powerful tool for systematic monitoring of geophysical phenomena over large areas potentially exposed to natural hazards. The Small BAseline Subset technique (SBAS, \citealp{berardino2002new}), in particular, stands out for its ability to derive time series of ground displacement with millimetric accuracy and at fine spatial resolution \citep{casu2006quantitative}. The SBAS algorithm builds a sequence of high-resolution images, each pixel of which contains the time series of ground displacement that occurred in the corresponding area of the ground with respect to a reference time. In the context of natural hazard assessment, SBAS-derived ground displacement data are routinely used to monitor subsidence, landslides, volcanic activity, and other geophysical processes. The increasing availability of such data has led to numerous applications in geosciences, civil protection, and environmental risk management (e.g. \citealp{poggi2026sentinel}; \citealp{di2021joint}; \citealp{lanari2010surface}).

While much of the focus has been put on studying mean and point deformation signals (e.g., \citealp{festa2022nation,monterroso2020global,bernardi2021use}), effective environmental prediction and monitoring require estimating the second-order structure of the data. At a basic level, the covariance informs uncertainty quantification in displacement maps and supports standard statistical tasks such as kriging or dimensionality reduction. More broadly, access to the full covariance enables scenario simulation, ground displacement reconstruction in unobserved areas, and change of support (e.g. \citealp{cressie2015statistics,banerjee2003hierarchical}).

However, retrieving the covariance structure of SBAS data poses several challenges. First, the data are inherently nonstationary due to the heterogeneity of the underlying geophysical phenomena. Second, information in large portions of the observed areas is missing: the scattering, absorption or reflection away from the sensor of radar signals -- happening when specific spatial entities on the ground are radiated, such as water, vegetation, or rocks -- cause the presence of missing values in some pixels of the SBAS-processed images. Additionally, the ground displacement measurements in those pixels are persistently
missing across all images. The persistent absence of observations for pairs of domain locations make standard covariance estimation infeasible for those locations.

By adopting a functional data analysis perspective (FDA; e.g., \citealp{ramsay2005functional}; \citealp{horvath2012inference}), we see the SBAS-processed images as functional data defined on a two-dimensional domain, acquired under a partial observation regime (e.g., \citealp{yao2005functional}; \citealp{kraus2015components}; \citealt{stefanucci2018PCA}; \citealp{descary2019recovering}; \citealt{kneip2020optimal}; \citealp{delaigle2021estimating} ). On the one hand, adopting the FDA perspective is convenient in that it offers a natural framework for modeling densely sampled processes characterized by intrinsic smoothness and continuity. On the other hand, this framework allows us to build on recent literature focusing on covariance estimation from partially observed functional data, to develop a new approach for covariance estimation specifically tailored to the regime of partial observation where entire portions of the domain are persistently missing over replicates.

\paragraph{Relation to prior work and main contribution}
The estimation of the covariance operator is a key task in FDA, since it underlies dimension-reduction approaches and provides an important tool for regularizing inference problems (\citealp{panaretos2013cramer}; \citealp{descary2019recovering}). In the case of fully observed functional data, consistent covariance estimators are readily available because the empirical covariance kernel can be observed across the entire domain. For partially observed data, however, the empirical estimator is available only at pairs of observed points, making direct estimation of the covariance infeasible for pairs involving the missing portions of the domain.

Existing work provides solutions to the covariance estimation problem under specific regimes of partial observation.
In the \emph{blanket regime} \citep{kraus2015components}, a sufficient number of functional samples covers the entire domain, so naive extensions of classical estimators are consistent and can be adopted directly.  
In the \emph{banded regime}~\citep{delaigle2021estimating,descary2019recovering}, each functional sample is observed on a subinterval of fixed length $\delta$. In this setting, the empirical covariance is consistently estimable only within a band of width $\delta$ around the diagonal, forming a so-called \emph{serrated domain}, and extrapolation to the missing portions is possible under a low-rank assumption on the true covariance \citep{descary2019functional,descary2019recovering}.  
In contrast, these methods fail when all samples are restricted to the same, possibly disjoint subdomain $O$. We refer to this setting as the \emph{fragmented regime}. Here, large regions of the domain are persistently unobserved across all replicates, so no consistent covariance estimates are available for pairs of points within or across those regions. Figure~\ref{fig:po-regimes} provides an illustration of the three regimes and of the key structural differences of the empirical covariance kernels in the three cases.

We propose a novel nonparametric approach to covariance estimation in the fragmented regime. Our method casts covariance estimation as a matrix completion problem, regularized by a Laplacian penalty that encourages smooth extension of the covariance kernel into unobserved areas while preserving flexibility in the observed regions. The proposed approach is fully nonparametric and does not rely on stationarity or isotropy assumptions.

\begin{figure}[ht]
    \centering
    \includegraphics[width=0.9\linewidth]{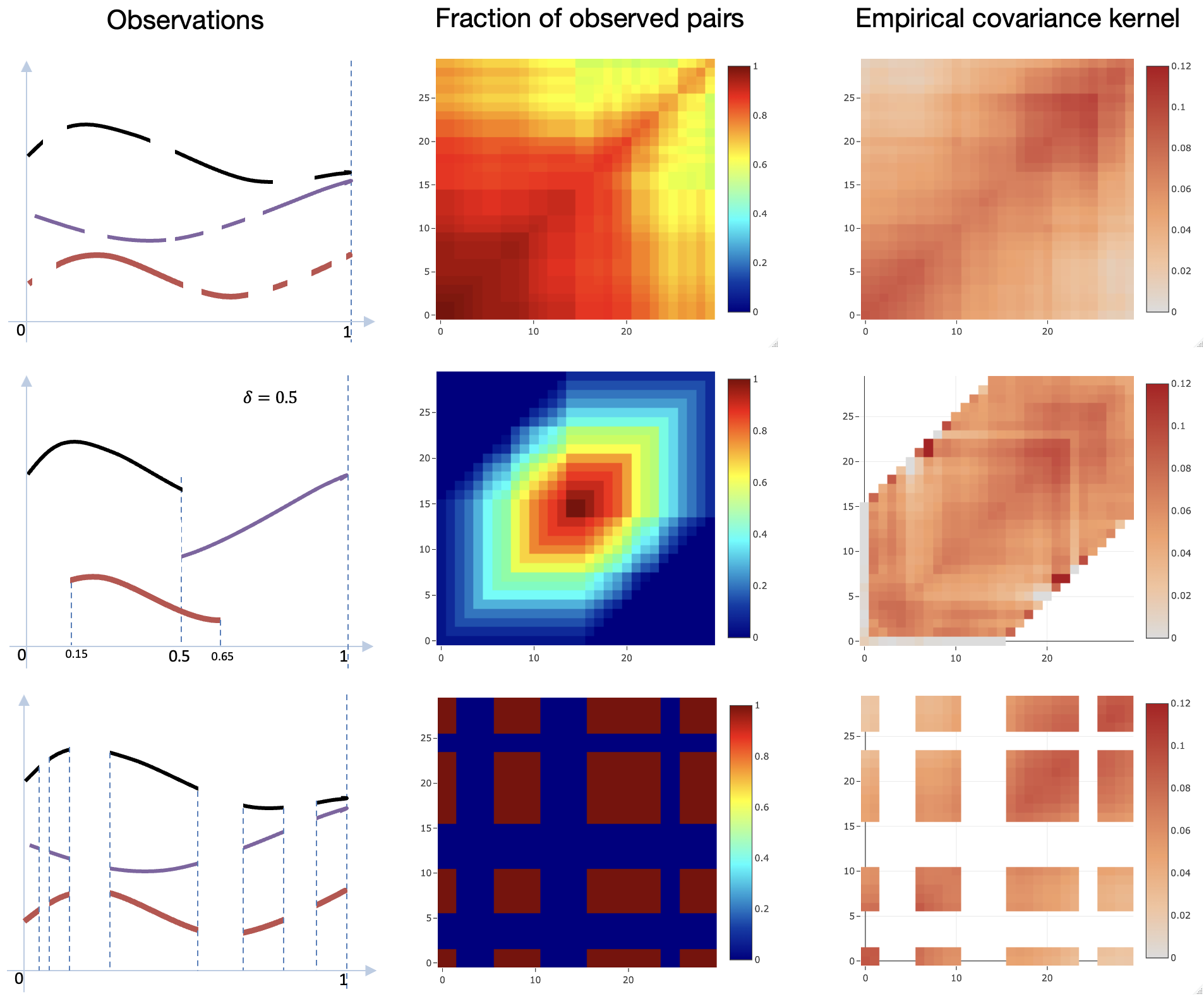}
    \caption{Illustration of the regimes of partial observation: {\em blanket regime, banded regime, fragmented regime}. The first column shows, for each regime, three functional data under that regime of partial observation. The second column shows, for each point $(s,t)$, the fraction of available observations out of a sample of $n=1000$ partially observed functional data. The last column displays the empirical kernel evaluated with the available observations at each point $(s,t)$.}
    \label{fig:po-regimes}
\end{figure}

We showcase the potential of our approach in terms of environmental monitoring from remote sensing data, by applying our method to estimate the second order structure of ground motion recorded by SBAS-processed SAR images relative to the area of the Phlegraean Fields, in central Italy, an area notably exposed to seismic and bradyseismic activity. 
Figure~\ref{fig:case-study-reconstructed-covar-pixels} anticipates our contribution in estimating the covariance from the available ground displacement data. The figure displays the estimated covariance between three selected locations (black dots) and all other points in the domain. Each column corresponds to one location: a central point where ground deformation is observed (first column), a point with missing observations (second column), and a location adjacent to the latter (third column). 
The top row reports the empirical covariance, which provides no information for pairs involving unobserved points (gray areas). 
The bottom row shows the covariance estimated with our method, which fills in the missing regions smoothly and extends the covariance structure from neighboring observed areas into the unobserved ones.

\begin{figure}[ht]
    \centering
    \includegraphics[width=0.8\linewidth]{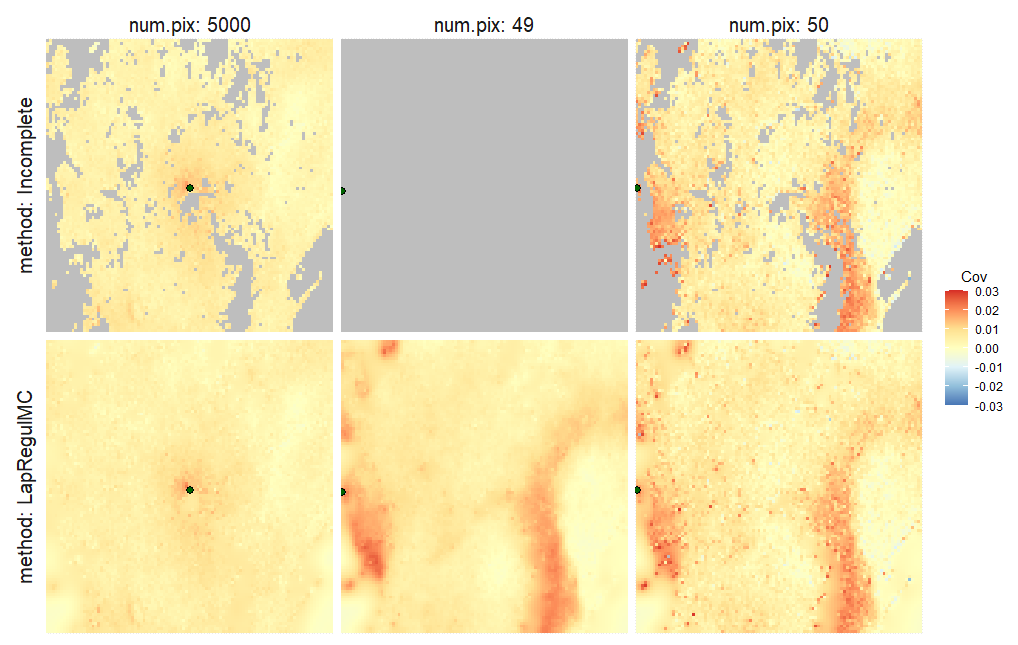}
    \caption{Covariance between three selected locations (black dots) and all other points in the domain. Columns correspond to: (i) a central location with observed ground deformation, (ii) a location with missing observations, and (iii) a location adjacent to the latter. The top row shows the empirical covariance, which is undefined for pairs involving unobserved points (gray areas). The bottom row shows the covariances estimated by our method.}
    \label{fig:case-study-reconstructed-covar-pixels}
\end{figure}

\paragraph{Outline of the paper}

The rest of the paper is organized as follows. In Section~\ref{sec:notation-and-problem-definition}, we introduce the notation and formally introduce the covariance estimation problem. In Section~\ref{sec:method}, we present our methodology for covariance operator estimation in the fragmented observation regime. Section~\ref{sec:simul-study} illustrates the performance of the proposed approach through a series of simulations, and identifies useful criteria for the selection of the  hyperparameters. After detailing how our regularized matrix completion problem can be extended to functional data with two-dimensional domain, we show in Section~\ref{sec:case-study} the estimated covariance structure of the SBAS-processed SAR images related to ground deformation monitoring. Then, we present potential byproducts of fully estimating the second-order structure of this data, including scenario generation and ground deformation reconstruction in the missing pixels. Finally, Section~\ref{sec:conclusions} concludes with a summary and a discussion of future directions.

The Appendices in the Supplementary Material provide additional details. Appendix~\ref{appendix:state-of-the-art} shows the performance of state-of-the-art methodologies for covariance estimation in the fragmented regime. A mathematical proof showing the well-posedness of our covariance estimation procedure in this regime is in Appendix~\ref{appendix:proof}. Further numerical results on synthetic experiments and on the interferometric data are in Appendices~\ref{appendix:simulation-study} and~\ref{appendix:case-study}.

\section{Notation and problem definition}
\label{sec:notation-and-problem-definition}

Let $\{X_i\}_{i=1}^n$ be independent and identically distributed random functions with values in $L^2([0,1])$, i.e.\ the Hilbert space of square-integrable functions defined on $[0,1]$. Under the assumption that $\mathbb{E}\left[ \lVert X_i \rVert_{L^2}^2 \right] < \infty$, we define the common mean function $\mu \colon [0,1] \rightarrow \mathbb{R}$ and covariance operator $\mathcal{R} \colon L^2([0,1]) \rightarrow L^2([0,1])$ of the random functional variables. As $\mathcal{R}$ is a Hilbert--Schmidt operator, it admits a kernel representation $r \colon [0,1]\times[0,1] \to \mathbb{R}$ such that
\begin{align*}
    \mathcal{R} f \; (t) = \int_0^1 r(s,t) f(s)\, ds, \qquad f \in L^2([0,1]).
\end{align*}
Recall that, for almost any $(s,t) \in [0,1]^2$, the covariance kernel reads $ r(s,t) = \operatorname{cov}(X_1(s), X_1(t)) $.

When $\{X_i\}_{i=1}^n$ are completely observed over $[0,1]$, the estimators of the mean and the covariance kernels are
\begin{equation}
     \hat{\mu}(s) = \frac{1}{n} \sum_{i=1}^n X_i(s), s \in [0,1], \qquad \hat{r}(s,t) = \frac{1}{n} \sum_{i=1}^n \left(X_i(s) - \hat{\mu}(s)\right) \left(X_i(t) - \hat{\mu}(t)\right), s,t \in [0,1].
     \label{eq:estimators_fully_observed}
\end{equation}
Under the assumptions of i.i.d.\ sampling, finite second moment, and complete observation of each trajectory, it is well known that $\hat{\mu}$ converges to $\mu$ in $L^2([0,1])$ and $\hat{r}$ converges to $r$ in $L^2([0,1]^2)$ as $n \to \infty$ (e.g.\ \citealp{bosq2000linear}). Thus, in the fully observed setting, the natural estimators in~\eqref{eq:estimators_fully_observed} are consistent.


We assume that $\{X_i\}_{i=1}^n$ are available only on the observable domain $ O \subset [0,1]$, which may in principle consist of multiple disjoint subintervals. Let $M := [0,1] \backslash O$ denote the unobserved part of the domain. As a consequence, $X_i(t)$ can be written as $X_i(t)=X_{iO}(t)\mathbf{1}_{O}(t)+X_{iM}(t)\mathbf{1}_{M}(t)$. As $X_{iM}(t)$ is unavailable, it is clear that the estimators in~\eqref{eq:estimators_fully_observed} are undefined whenever $t \in M$ or $s \in M$.

A naive extension of the fully observed estimators for the fragmented regime is obtained by simply setting to zero the contributions outside $O$. In analogy with patched estimators for partially observed functional data~\citep{kraus2015components,descary2019functional}, define
\begin{align}
    \hat{\mu}(s) &= \frac{\I{s \in O}}{n} \sum_{i=1}^n X_i(s), \label{eq:mean_estimator_fragmented_regime}\\
    \hat{r}(s,t) &= \frac{\I{s \in O} \cdot \I{t \in O}}{n} \sum_{i=1}^n (X_i(s) - \hat{\mu}(s))(X_i(t) - \hat{\mu}(t)).
    \label{eq:kernel_estimator_fragmented_regime}
\end{align}
These estimators coincide with those in~\eqref{eq:estimators_fully_observed} when $s,t \in O$, and are equal to zero otherwise. They inherit consistency from~\eqref{eq:estimators_fully_observed} on pairs of locations within the observed domain. Whenever $\I{s \in O} \cdot \I{t \in O}=0$, on the other hand, the estimator~\eqref{eq:kernel_estimator_fragmented_regime} sets all covariances involving unobserved locations to zero, and most critically assigns zero variance to unobserved points. While setting missing covariances to zero generally results in a distorted representation of the second-order structure of the data, assigning zero variance to unobserved locations leads to a severe underestimation of the variability of the data generating process.

State-of-the-art methods for covariance estimation from partially observed functional data \citep{descary2019recovering,delaigle2021estimating} proceed by extrapolating~\eqref{eq:kernel_estimator_fragmented_regime} into unobserved regions. Both approaches guarantee consistency under the assumption that the empirical covariance kernel is available on a {\em serrated domain}, i.e., within a band of width $\delta$ around the main diagonal. Figure~\ref{fig:reconstruction-state-of-the-art} illustrates the performance of the two methods in the fragmented regime, where this assumption is clearly violated.
The results show that neither method is able to recover a meaningful covariance structure in this setting.

Figures~\ref{fig:exp2_alpha0.1_m1_K30_rank20_reduced}–\ref{fig:exp4_alpha0.1_m1_K30_rank20_reduced} in Appendix~\ref{appendix:state-of-the-art} confirm that this failure is not specific to the missing-data pattern used in the first experiment. In fact, the results clearly show that the tendency to set covariances to zero becomes increasingly severe for both methods as the width of the missing intervals increases.

\begin{figure}[ht]
    \centering
    \includegraphics[width=1\linewidth]{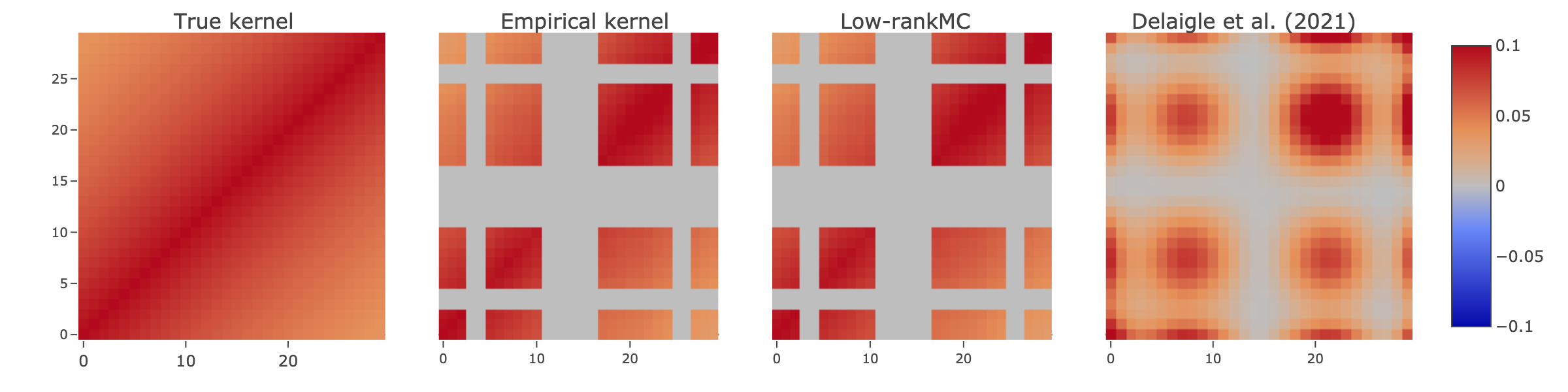}
    \caption{Example of covariance kernel estimation with state-of-the-art methods for partially observed functional data in the fragmented regime. From the left: true covariance kernel; empirical covariance kernel~\eqref{eq:kernel_estimator_fragmented_regime}; covariance kernel estimation with the method of~\citet{descary2019recovering} ({\em Low-rankMC}); covariance kernel estimation with the method of~\citet{delaigle2021estimating}. The empirical covariance kernel is obtained from a sample of $n=1000$ functional observations of a Gaussian process with stationary covariance defined by the Matérn kernel. Each functional sample is observed on a regular grid of $K=30$ points. The fragmented regime is simulated by selecting $J_M = \{4,5,12:17,26,27\}$ as the set of missing indexes along the grid (see Section~\ref{sec:data-sim}).}
    \label{fig:reconstruction-state-of-the-art}
\end{figure}

These preliminary analyses highlight the need for improved covariance estimation in the fragmented regime. As for the mean function, our covariance estimator will require knowledge of $\mu$ only on the observed domain $O$. Reconstructing $\mu$ over the entire domain, although possible given a covariance estimate through trajectory reconstruction methods \citep{kraus2015components,kneip2020optimal,kraus2020ridge} or geostatistical techniques such as universal kriging \citep{cressie2015statistics,chiles2012geostatistics}, lies beyond the scope of this work.

\paragraph{Discretization of the estimation problem.}
Since this setting naturally fits the real data analyzed in this work, as well as most practical applications of FDA, we focus on the case where functional data are observed on a finite grid of points, in line with the work of~\citet{descary2019recovering}. Specifically, we assume that each trajectory is observed at $K \geq 2$ locations $(t_1, \dots, t_K) \in T \subset O$, corresponding to consecutive points along the domain. Given an equally spaced partition $\{I_j\}_{j=1}^K$ of the interval $[0,1]$, each function $X_i$ is then represented by the values $X_i(t_1), \dots, X_i(t_K)$, with $t_j \in I_j$ for all $j=1,\dots,K$. This leads to the following piecewise constant, $K$-resolution approximation of $X_i$:
\begin{equation}
    X_i^K(t) = \sum_{j=1}^K X_i(t_j) \, \I{t \in I_j}.
    \label{eq:emp-data}
\end{equation}
The $K$-resolution version of the true covariance kernel is
\begin{equation}
    r^K(s,t)=\operatorname{Cov}(X_i^K(s),X_i^K(t))
    =\sum_{j,l=1}^K r(t_j,t_l) \, \I{(s,t) \in I_j \times I_l},
    \label{eq:cov_disc}
\end{equation}
and its empirical counterpart is obtained by replacing $r(t_j,t_l)$ in~\eqref{eq:cov_disc} with the estimator $\hat{r}(t_j,t_l)$ from~\eqref{eq:kernel_estimator_fragmented_regime}, i.e.,
\begin{equation}
    \hat{r}^K(s,t) = \sum_{j,l=1}^K \hat{r}(t_j,t_l) \, I{(s,t) \in I_j \times I_l}.
    \label{eq:cov_disc_emp}
\end{equation}

Throughout this work, we represent~\eqref{eq:cov_disc_emp} via the empirical covariance matrix $\hat{R}^K = \{\hat{r}(t_j,t_l)\}_{j,l=1}^K$.
In the discretized setting, the fragmented regime implies that each trajectory is observed only on a subset $\mathcal{J}_O$ of grid indices, while it is missing on $\mathcal{J}_M = [K] \backslash \mathcal{J}_O$.
Let $\Omega_O := \mathcal{J}_O \times \mathcal{J}_O$ denote the set of observed index pairs, and let $\Omega_M  = \left([K]\times[K] \right) \backslash \Omega_O$. In accordance to~\eqref{eq:cov_disc_emp}, all entries of $\hat{R}^K$ corresponding to $\Omega_M$ are equal to zero.
In the following, we cast the problem of estimating the covariance operator as a regularized matrix completion problem. The estimated entries of the matrix serve as coefficients in~\eqref{eq:cov_disc_emp}, which in turn defines an estimator of the continuous covariance kernel. Note that the piecewise constant approximation in~\eqref{eq:cov_disc_emp} is only one possible discretization strategy. If one sought a smoother solution, one could instead expand the kernel in a B-spline basis or apply a kernel smoothing procedure to~\eqref{eq:cov_disc_emp} as a post-processing step (see, e.g., \citealp{ramsay2005functional}).

\section{Method}
\label{sec:method}
We view the estimation of the covariance operator as a matrix completion problem for the empirical covariance matrix. This leads to a finite-dimensional optimization problem, where the unknown is the covariance matrix itself, identified by minimizing an objective functional composed of a misfit term and a Laplacian penalty. The latter encourages smooth transitions between neighboring entries of the covariance operator. The assumption underlying the methodology is that, in the unobserved regions, the second order structure of the data-generating process can be retrieved as a smooth extension of the structure observed at the boundaries of the missing regions.



\subsection{Covariance estimation via regularized matrix completion}
\label{sec:cov-estimation}
Let $\hat{R}^K$ denote the $K$-resolution empirical covariance matrix introduced in the previous section. 
Let $\mathbb{S}_{+}^K \subset \mathbb{R}^{K\times K}$ be the set of $K\times K$ symmetric positive semidefinite matrices. We estimate the true covariance matrix by solving

\begin{equation}
\min_{\theta\in \mathbb{S}_{+}^K}\;
\sum_{(i,j)\in \Omega_O} \left(\hat R^K_{ij}-\theta_{ij}\right)^2
\;+\;
\alpha\!\left[
m\!\!\sum_{(i,j)\in \Omega_M}\!\! [L \theta]_{ij}^2
+
(1-m)\!\!\sum_{(i,j)\in \Omega_O}\!\! [L \theta]_{ij}^2
\right],
\label{eq:min-pb-ij}
\end{equation}
where $L$ is a discretized Laplacian operator acting along the rows and columns of $\theta$, as further detailed in Subsection \ref{sec:Laplacian}. The coefficient $\alpha>0$ controls the trade-off between fidelity to the observed entries and smoothness of the resulting covariance estimate. The parameter $m\in(0,1]$ controls the relative influence of the Laplacian penalty in the missing and in the observed parts of the matrix, thereby interpolating between two estimation strategies. When $m=1$, the observed entries are treated as reliable anchor points and the regularization term is used primarily to extend their structure into the missing regions. Values of $m$ close to $0$, on the other hand, apply regularization mainly where data are observed, denoising the empirical covariance and stabilizing the matrix-completion problem.
Given the central role of $\alpha$ and $m$, it is crucial to identify practical, data-driven criteria for their selection; we propose such criteria based on an extensive simulation study, which we discuss in Section~\ref{sec:simul-study}.

\subsection{Laplacian regularization}\label{sec:Laplacian}
The discretized Laplacian $L$ acts differently on diagonal and off-diagonal entries of the covariance matrix. This separation ensures that missing variances are inferred only from neighboring variances along the main diagonal, while missing covariances are inferred from neighboring covariances.

Let $L_1$ denote the one-dimensional discrete Neumann Laplacian; that is, the $K \times K$ matrix
\begin{equation*}
L_1 =
\begin{bmatrix}
1 & -1 & 0 & \cdots & 0 & 0 \\
-1 & 2 & -1 & \ddots & \vdots & \vdots \\
0 & -1 & 2 & \ddots & 0 & 0 \\
\vdots & \ddots & \ddots & \ddots & -1 & 0 \\
0 & \cdots & 0 & -1 & 2 & -1 \\
0 & \cdots & 0 & 0 & -1 & 1
\end{bmatrix}.
\end{equation*}
For off-diagonal entries, we enforce smoothness in both row and column directions, with a discrete two-dimensional Laplacian:
\begin{equation}
    [L \theta]_{ij} = [L_1 \theta + \theta L_1]_{ij} = -\theta_{i-1,j} - \theta_{i+1,j} - \theta_{i,j-1} - \theta_{i,j+1} + 4\theta_{ij},
    \label{eq:Lij}
\end{equation}
with adaptations at the boundaries.
The Neumann Laplacian corresponds to a specific choice of boundary adaptation. That is, for $i=1$ or $i=K$, the second difference is computed using the available two points, which corresponds to a one-sided finite-difference approximation of the second derivative. This adaptation is only one of several possible choices concerning treatments of the boundaries that can be borrowed from standard numerical analysis, such as reflective, periodic, or higher-order finite-difference schemes (e.g., \citealp{quarteroni2006numerical}; \citealp{leveque2007finite}). 

For diagonal entries, we penalize deviations from smoothness along the main diagonal -- i.e.\ across neighboring variances. That is,
\begin{equation}
    [L \theta]_{ii} = [L_1 \mathrm{diag}(\theta)]_{i} = -\theta_{i-1,i-1} + 2\theta_{ii} - \theta_{i+1,i+1},
    \label{eq:Lii}
\end{equation}
again with boundary adaptations obtained by replacing missing neighbors with one-sided differences. In~\eqref{eq:Lii}, $\mathrm{diag}(A)$ denotes the vector of diagonal elements of any square matrix $A$.

The discretized Laplacian $L$ discourages abrupt local changes separately along the diagonal and the off-diagonal entries of $\theta$. An intuition behind this type of regularization is the following. Provided that $m$ is strictly smaller than 1, in the extreme case when the regularization parameter $\alpha$ in~\eqref{eq:min-pb-ij} becomes very large the solution tends toward a particularly simple form in which all variances become equal and all covariances converge to a common value, resembling a nugget-like structure.

In the next subsection, we rely on the observation that the two distinct regularizations for diagonal and off-diagonal elements of the covariance matrix can be obtained with the discrete one-dimensional Neumann Laplacian $L_1$ and by introducing suitable masks, as is detailed next. Additionally, we make the modelling choice of not calculating the Laplacian penalty in the super- and sub-diagonals, i.e.\ we avoid using variances to smooth covariances. This design choice prevents underestimation of variability in the presence of possible nugget effects, where variances may exhibit sharp discontinuities relative to nearby covariances. This too is achievable by suitably specifying the masks entering the optimization functional.

\subsection{Optimization problem in matrix form}
\label{sec:optimization-matrix-form}
It is easy to show that problem~\eqref{eq:min-pb-ij} can be expressed in matrix form by introducing suitable masks and weighting matrices that rule the interplay between the observed and the missing cells of the covariance matrix. 

Let $I$ be the $K \times K$ identity matrix, and $\bar{I}$ be the $K \times K$ mask taking value 0 in the diagonal and the first super- and sub-diagonals, and 1 elsewhere. Let $J$ denote the $K \times K$ matrix of all ones.
Let $P^O$ denote a $ K \times K $ mask such that $P^O_{ij} \;=\; \I{(i,j) \in \Omega_O}$. Let $P^m = \sqrt{m} P^O + \sqrt{1-m} (J - P^O)$. The quadratic optimization problem in matrix form then reads
\begin{align*}
    & \underset{\theta \in \mathbb{S}_{+}^K}{\text{min}} \left\{ 
        \left\| P^O \circ (\hat{R}^K - \theta) \right\|_F^2 \right. \\
    & \left. \qquad 
        + \; \alpha  \; \left( 
            \left\| P^m \circ \bar{I} \circ (L_1\theta + \theta L_1) \right\|_F^2 
            + \left\| P^m \circ I \circ \left( L_1 (I \circ \theta) J \right) \right\|_F^2 
        \right) 
    \right\}, \numberthis \label{eq:min-pb-theta}
\end{align*}
where $\|\cdot\|_F$ denotes the Frobenius norm.

\begin{proposition} \label{prop:wellposed}
The optimization problem~\eqref{eq:min-pb-theta} admits a unique minimizer $\theta^\star\in\mathbb{S}_+^K$.
\end{proposition}

The proof of Proposition~\ref{prop:wellposed} is deferred to Appendix~\ref{appendix:proof}. To ensure that the solution to problem~\eqref{eq:min-pb-theta} is a covariance, we reparametrize the problem by leveraging the fact that any symmetric and positive semidefinite matrix can be factorized as $\theta = \gamma \gamma^\top$, where $\gamma \in \mathbb{R}^{K \times K}$. This change of variables eliminates the need to explicitly enforce symmetry and positive definiteness, allowing us to recast~\eqref{eq:min-pb-theta} as an unconstrained problem:

{\small
\begin{align*}
    &\underset{\gamma \in \mathbb{R}^{K \times K}}{\text{min}} \left\{ 
        \left\| P^O \circ (\hat{R}^K - (\gamma\gamma^T)) \right\|_F^2 \right.\\
        &\left. \qquad
        + \; \alpha  \; \left( 
            \left\| P^m \circ \bar{I} \circ (L(\gamma\gamma^T) + (\gamma\gamma^T) L) \right\|_F^2 
            + \left\| P^m \circ I \circ \left( L (I \circ (\gamma\gamma^T)) J \right) \right\|_F^2 
        \right) 
    \right\}. \numberthis \label{eq:min-pb-gamma}
\end{align*}}

For a specified choice of the parameters $\alpha$ and $m$, the optimization problem is solved using the BFGS algorithm -- an iterative Quasi-Newton method for solving unconstrained nonlinear optimization problems.
Note that the mapping $\gamma \to \gamma \gamma^T$ entails orthogonal invariances ($\gamma Q$ and $\gamma$ yield the same $\gamma \gamma^T$ for any orthogonal $Q$), which introduces multiple equivalent representations of the same solution and breaks the convexity of the problem. Although the reformulated objective is nonconvex in $\gamma$, BFGS exhibits stable convergence in practice and often yields high-quality solutions efficiently.


We initialize the matrix $\gamma$ using the singular value decomposition of the empirical covariance matrix, where we assign value zero to the missing cells. From synthetic experiments, we noticed that convergence typically speeds up if the missing cells in the empirical covariance matrix in the diagonal and the first super- and sub-diagonal in $\hat{R}^K$ are linearly interpolated with the observed values along that diagonal. Let $\hat{R}_0^K$ denote the initialization of the $K$-resolution covariance matrix. Then, $\gamma_0 = U\Lambda^{1/2}$, where $U$ and $\Lambda$ are such that $\hat{R}_0^K = U\Lambda U^\top$.

\section{Simulation study}
\label{sec:simul-study}
With the synthetic experiments reported in this section, our purposes are (i) to assess the quality of the proposed covariance estimator against other existing approaches for covariance estimation, and (ii) to identify practical, data-driven criteria for selecting the hyperparameters $(\alpha,m)$ in~\eqref{eq:min-pb-ij}.

\subsection{Data generation}
\label{sec:data-sim}
We generate $n$ functional data from a zero-mean Gaussian process on $[0,1]$ with a given covariance kernel. Each trajectory is discretized on a uniform grid of $K=50$ points, so that it can be represented as in~\eqref{eq:emp-data}. To simulate the fragmented regime of partial observation, we fix a subset $J_M \subset [K]$ of indices corresponding to missing data, so that the simulated trajectories are unobserved on $30 \%$ of the points of the grid. As the indices of $J_M \subset [K]$ are sampled randomly, some of them may be contiguous generating gaps of missing information. The mask $P^O\in\{0,1\}^{K\times K}$, introduced in the previous section, is defined so that $P^O_{ij} = 0$ if $i \in J_M$ or $j \in J_M$, and 1 otherwise. The mask is applied to the empirical covariance matrix computed from the $n$ discretized trajectories, yielding partially observed covariance matrices.

We consider two main covariance structures: a stationary and isotropic Matérn model $R_S$ and a fully nonstationary and anisotropic covariance $R_{NS}$. The stationary covariance is defined by the Matérn kernel
\begin{align*}
    r_S(d) \;=\; \sigma^2 \cdot \frac{2^{1-\nu}}{\Gamma(\nu)} \cdot \Big(\tfrac{d}{\phi}\Big)^\nu \cdot K_\nu\!\Big(\tfrac{d}{\phi}\Big),
\end{align*}
with $d$ the distance between grid locations, $\nu>0$ smoothness, $\phi>0$ range, and $K_\nu$ the modified Bessel function of the second kind. We set $(\sigma^2,\nu,\phi)$ as in the illustrations, and form $R_S$ by evaluating $r_S$ on the $K$-point grid.
To realistically represent a process with nonstationarity and anisotropy in its second order structure, we construct $R_{NS}$ from a $50\times 50$ block extracted from an empirical covariance surface estimated on real data (see Section~\ref{sec:case-study}). This yields heterogeneous local variances and direction-dependent correlations.

\subsection{Evaluation metrics}
\label{sec:metrics}
Let $\bar{R}^K = {\bar{r}_{ij}}$ denote a reference covariance matrix, resulting from the discretization of a reference covariance kernel $\bar{r}$ at $K$-resolution, in accordance to~\eqref{eq:cov_disc}. Let $\hat{\theta} = \hat{\gamma} \hat{\gamma}^T$, where $\hat{\gamma}$ denotes the solution to~\eqref{eq:min-pb-gamma}. We define the total, observed-only, and missing-only root mean squared errors with respect to the reference covariance matrix as

{\small
\begin{align}
    \mathrm{RMSE}_{\mathrm{tot}}(\bar{R}^K,\hat{\theta}) 
    &= \frac{1}{\sqrt{K^2}}\lVert \bar{R}^K - \hat{\theta} \rVert_F = \Bigg(\frac{1}{K^2}\sum_{i,j=1}^K \big(\bar{r}_{ij}-\hat{\theta}_{ij}\big)^2 \Bigg)^{1/2}, \label{eq:rmse-tot}\\[2mm]
    \mathrm{RMSE}_{O}(\bar{R}^K,\hat{\theta}) 
    &= \frac{1}{\sqrt{N_O}}\lVert P^O \circ ( \bar{R}^K - \hat{\theta}) \rVert_F = \Bigg(\frac{1}{N_{O}}\sum_{i,j=1}^K P^O_{ij}\,\big(\bar{r}_{ij}-\hat{\theta}_{ij}\big)^2 \Bigg)^{1/2}, \label{eq:rmse-obs}\\[2mm]
    \mathrm{RMSE}_{M}(\bar{R}^K,\hat{\theta}) 
    &= \frac{1}{\sqrt{N_M}}\lVert (J - P^O) \circ ( \bar{R}^K - \hat{\theta}) \rVert_F = \Bigg(\frac{1}{N_{M}}\sum_{i,j=1}^K (1-P^O_{ij})\,\big(\bar{r}_{ij}-\hat{\theta}_{ij}\big)^2 \Bigg)^{1/2}, \label{eq:rmse-miss}
\end{align}}

where $N_O=\sum_{ij} P^O_{ij}$ and $N_M=K^2-N_O$ denote the number of observed and missing entries, respectively.

In the synthetic experiments reported in Section~\ref{sec:comp-methods}, we compare the performance of our estimation method with that of other state-of-the-art methods by considering as reference covariance matrix $\bar{R}^K$ both the true covariance matrix $R^K$ and the empirical covariance matrix $\hat{R}^K$, the latter obtained from a sample of $n$ synthetic functional data.

Since the true covariance kernel is unobservable in real applications, it is natural to rely on errors~\eqref{eq:rmse-tot}--\eqref{eq:rmse-miss} computed with respect to the empirical covariance $\hat{R}^K$. This is supported by well-known theoretical results showing that $\hat{R}^K$ is the data-driven approximation of the true covariance and concentrates around it under mild conditions (e.g., \citealp{yao2005functional, johnstone2009consistency, descary2019recovering, delaigle2021estimating}). Consequently, errors computed with respect to $\hat{R}^K$ provide a natural proxy for the errors computed with respect to $R^K$.
In Section~\ref{sec:increasing_nsamples}, we complement this theoretical perspective with a numerical study investigating this behaviour for increasing sample sizes. Building on these insights, Section~\ref{sec:hyperparam-sel} introduces practical criteria for hyperparameter selection based on the available parts of the empirical covariance.

\subsection{Comparison with existing estimation methods}
\label{sec:comp-methods}
We compare our estimator ({\em LapRegulMC}) to two alternative methods, namely the low-rank matrix completion method proposed in~\citet{descary2019recovering} ({\em Low-rankMC}), and a parametric geostatistical estimator ({\em Geostationary}). The latter is obtained by fitting the parameters $(\sigma^2,\nu,\phi)$ of the Matérn covariance via least squares evaluated in the observed entries of the empirical covariance.

We do not include a comparison with the method of~\citet{delaigle2021estimating}. As anticipated in Section~\ref{sec:notation-and-problem-definition}, the main distinction between this approach and that of~\citet{descary2019recovering} lies in the adopted representation of the functional data. The former relies on a basis expansion, while the latter is developed in a discretized setting. Beyond this technical difference, however, both methods rely on the same working assumption of {\em serrated domain} and lead to very similar results (see~\citet{delaigle2021estimating}). For this reason, we considered it redundant to include both in the comparison.

The methods are compared across five different covariance structures with varying degrees of anisotropy and nonstationarity. This is achieved by considering a sequence of covariance structures indexed by $p \in [0,1]$ that move from the stationary and isotropic Matérn model $R_S$ ($p=1$), to the fully nonstationary and anisotropic covariance $R_{NS}$ ($p=0$). For intermediate $p$, we define the population covariance matrix as a convex combination of the two structures, i.e.
\begin{align*}
    R(p) \;=\; p\, R_S^* \;+\; (1-p)\, R_{NS}^*,
\end{align*}
where $R_S^*$ and $R_{NS}^*$ are scaled versions of $R_S$ and $R_{NS}$ having the same trace, so that overall marginal variance is comparable across $p$. Specifically, we let $p$ vary in the set $\{0, 0.25, 0.5, 0.75, 1\}$.

For each $p$, we do a Monte Carlo simulation with $30$ replicates, each time generating $n=200$ functional data. Since some of the competing methods require hyperparameter tuning that may depend on the specific pattern of missing indices $\mathcal{J}_M$, we draw $J_M$ at random once for each $p$ and keep it fixed across all replicates.

The {\em Geostationary} method does not involve hyperparameters.
For {\em Low-rankMC}, a target rank must be specified. Following~\citet{descary2019recovering}, we determine the rank by visual inspection of $\mathrm{RMSE}_{\mathrm{obs}}(\hat{R}^K, \hat{\theta})$, identifying the elbow point. The inspection is performed once for each $p$, yielding the ranks ${10,7,5,4,3}$ for $p \in \{0,0.25,0.5,0.75,1\}$, respectively. The ranks are then held fixed across all replicates.
For {\em LapRegulMC}, the hyperparameters at this stage are selected automatically via a grid search, retaining the pair $(\alpha, m)$ that minimizes $\mathrm{RMSE}_{\mathrm{tot}}(\hat{R}^K, \hat{\theta})$. To ensure a fair comparison with {\em Low-rankMC}, hyperparameter selection is likewise performed only once for each $p$ and the resulting $(\alpha, m)$ are used consistently across all replicates.

The results are synthesized by the boxplots in Figure~\ref{fig:methods-comparison}, which compare the errors $\mathrm{RMSE}_{\text{tot}}(R^K, \hat{\theta})$ achieved by all three methods across the Monte Carlo replicates and for each $p$.
As expected, the {\em Low-rankMC} method fails to recover the covariance structure in all scenarios. This confirms that the serrated-domain assumption is indeed a necessary condition for the validity of the method, and cannot be relaxed without severe loss of accuracy.
When $p=1$, the {\em Geostationary} estimator serves as a natural benchmark and shows that our method performs competitively with a correctly specified parametric model. As the covariance structure departs from the Matérn form (smaller $p$), the performance of {\em Geostationary} deteriorates, whereas the estimation error of our method remains consistently low across all values of $p$. Figure~\ref{fig:methods-comparison-emp} in Appendix~\ref{appendix:comparison-with-existing-methods} shows analogous results for the total estimation errors evaluated with respect to the empirical covariance. 
Overall, these experiments highlight the effectiveness of our method in accurately capturing the second-order structures, regardless of the degree of stationarity or isotropy in the data generating process.

\begin{figure}[ht]
    \centering
    \includegraphics[width=0.8\linewidth]{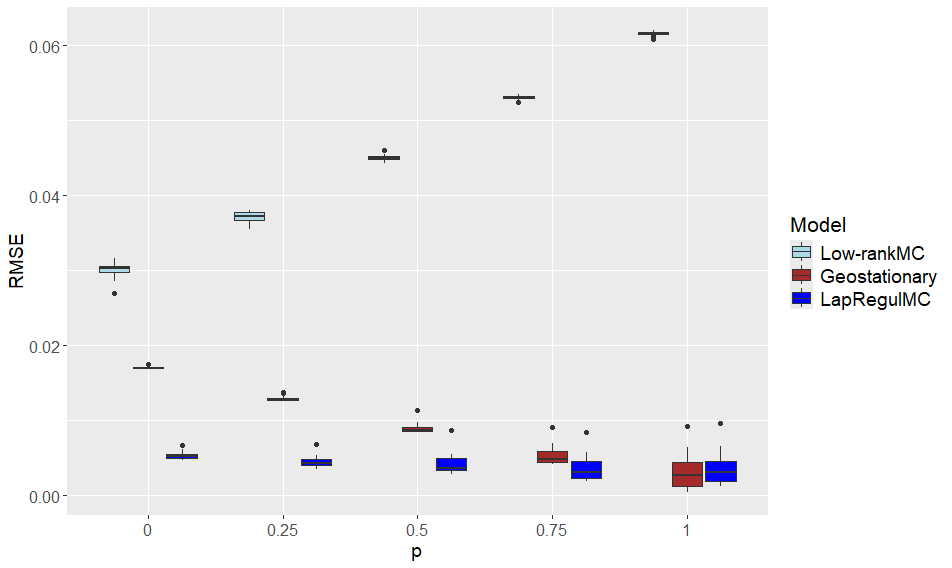}
    \caption{Comparison of the performance in terms of $\text{RMSE}_{\mathrm{tot}}(R^K, \hat{\theta})$ for the three models for covariance estimation: the method proposed in \cite{descary2019recovering} {\em (Low-rankMC)}, a parametric geostationary method {\em(Geostationary)} and our method {\em (LapRegulMC)}.}
    \label{fig:methods-comparison}
\end{figure}

\subsection{The impact of relying on finite sample size}
\label{sec:increasing_nsamples}
A detailed simulation study reported in Appendix~\ref{appendix:increasing_nsamples} examines the impact of selecting the hyperparameters based on the empirical covariance matrix from a finite sample, rather than the true covariance, on the performance of the proposed estimation routine. The results show that the hyperparameter selected by minimizing the error with respect to the empirical covariance yield estimation performance essentially identical to that obtained using the (infeasible) true covariance, even for moderate sample sizes (e.g., $n=100$). Moreover, the distributions of the hyperparameters selected by minimizing $\mathrm{RMSE}_{\mathrm{tot}}(\hat{R}^K, \hat{\theta})$ or $\mathrm{RMSE}_{\mathrm{tot}}(R^K, \hat{\theta})$ become increasingly similar as the sample size grows. This analysis supports the identification of practical criteria for hyperparameter selection based on the estimation error with respect to the empirical covariance, which we discuss next.

\subsection{Practical criteria for hyperparameter selection}
\label{sec:hyperparam-sel}
We conduct a Monte Carlo simulation with $B=30$ replicates, each generating $n=200$ functional data following the procedure in Section~\ref{sec:data-sim}. As the choice of the hyperparameters is expected to depend on the pattern of missing values, the subset $J_M$ of missing indices is sampled randomly and kept fixed across the MC replicates. For each partially observed empirical covariance matrix and for a fixed pair $(\alpha,m)$, we estimate the full covariance matrix by applying the procedure described in Section~\ref{sec:optimization-matrix-form}, and compute the estimation errors~\eqref{eq:rmse-tot}–\eqref{eq:rmse-miss} evaluated with respect to the empirical covariance $\hat{R}^K$. In the following, $\mathrm{RMSE}_{\bullet}$ is short for $\mathrm{RMSE}_{\bullet} (\hat{R}^K, \hat{\theta})$.

\paragraph{Selection of parameter $\alpha$.}
With the parameter $m$ fixed, i.e. $m=0.6$, we inspect the behavior of the errors as function of $\log_{10}\alpha$. Figure~\ref{fig:param-alpha-stat} reports the functional boxplots of the resulting errors under the isotropic and stationary covariance structure introduced in Section~\ref{sec:data-sim}.
In all three panels, the dashed vertical line marks the value of $\alpha$ minimizing the mean of $\mathrm{RMSE}_{\mathrm{tot}}$ across replicates. A first observation is that the minimizer of $\mathrm{RMSE}_{\mathrm{tot}}$ is close to the minimizer of $\mathrm{RMSE}_{M}$. Although in practical applications neither $\mathrm{RMSE}_{\mathrm{tot}}$ nor $\mathrm{RMSE}_{M}$ are directly observable, our goal is to choose $\alpha$ yielding a estimation error close to the minimum of $\mathrm{RMSE}_{\mathrm{tot}}$.

The role of $\alpha$ in the objective implies that $\mathrm{RMSE}_{O}$ grows monotonically with $\alpha$: as $\alpha$ increases, the estimator departs from an exact fit to the observed entries and encourages smoother solutions, thereby increasing the error relative to the observed covariance.
Interestingly, the $\alpha$ that minimizes $\mathrm{RMSE}_{\mathrm{tot}}$ consistently corresponds to the elbow point of the $\mathrm{RMSE}_{O}$ curve, i.e., the largest $\alpha$ before $\mathrm{RMSE}_{O}$ undergoes a sustained increase. The same pattern holds under the nonstationary covariance model of Section~\ref{sec:data-sim} (see Appendix~\ref{appendix:simulation-study}, Figure~\ref{fig:param-alpha-nonstat}), and for different values of $m$ (see Appendix~\ref{appendix:simulation-study}, Figures~\ref{fig:stat-paramsvarying} and~\ref{fig:nonstat-paramsvarying}, left panels).
These findings suggest a simple and effective practical rule: select $\alpha$ at the elbow of the $\mathrm{RMSE}_{O}$ curve. This criterion provides a data-driven proxy for minimizing $\mathrm{RMSE}_{\mathrm{tot}}$, even in the absence of ground-truth knowledge, and provides low total estimation error in practice.

\begin{figure}[ht]
    \centering
    \includegraphics[width=0.9\linewidth]{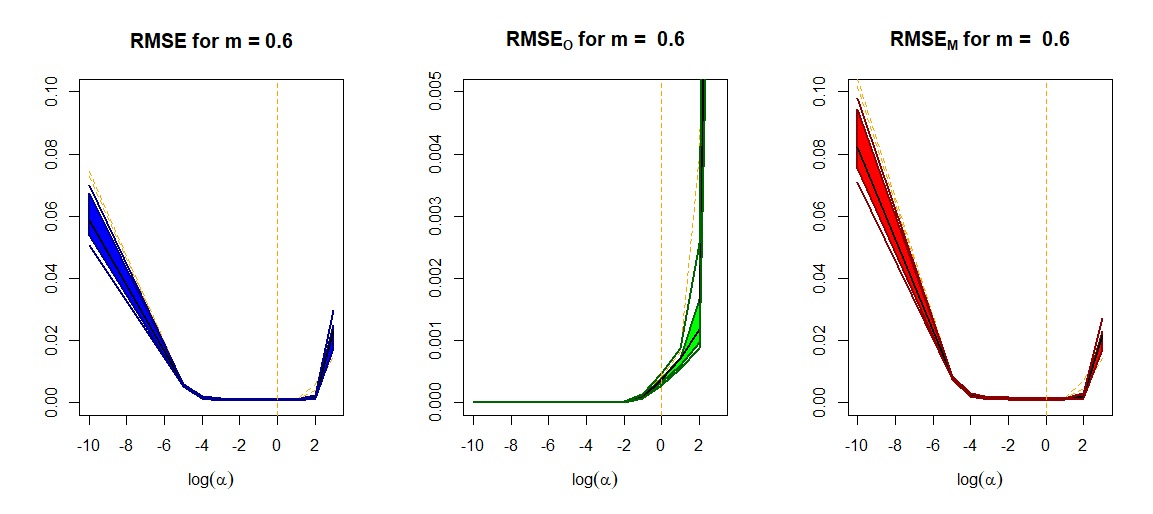}
    \caption{Functional boxplots of the curves of RMSE generated per each value of the logarithm of $\alpha$ over 30 simulations, for $m=0.6$ in the stationary scenario. Each functional boxplot is related to a different error: $\mathrm{RMSE}_{\mathrm{tot}}$ (blue), RMSE$_O$ (green), RMSE$_M$ (red). The point $log(\alpha)=-10$ on the x-axis serves as a fictitious representation graphically introduced to illustrate the value of RMSE for $\alpha = 0$. The black curve indicates the median error over all replicates. In all three plots, the dashed vertical line is in correspondence the minimum of the mean of $\mathrm{RMSE}_{\mathrm{tot}}$ over all replicates. In the plot of $\mathrm{RMSE}_{O}$, the vertical line indicates the identified criterion for selecting $\alpha$. The dashed orange curves are outliers of the MC sample,  identified by the functional boxplots.}
    \label{fig:param-alpha-stat}
\end{figure}


\paragraph{Selection of parameter $m$.}

With $\alpha$ set following the criterion identified above, we now investigate the effect of $m$ on the estimation error. Figure~\ref{fig:param-m-stat} reports the boxplots of the errors under the isotropic and stationary covariance model of Section~\ref{sec:data-sim}. As before, the dashed vertical line indicates the value of $m$ minimizing the mean of $\mathrm{RMSE}_{\mathrm{tot}}$ across replicates. We observe that the minimizer of $\mathrm{RMSE}_{\mathrm{tot}}$ is close to the minimizer of $\mathrm{RMSE}_{M}$, confirming that the latter provides a useful proxy for the former. Since neither $\mathrm{RMSE}_{\mathrm{tot}}$ nor $\mathrm{RMSE}_{M}$ are available in practice, our goal is to select $m$ that approximates the minimizer of $\mathrm{RMSE}_{\mathrm{tot}}$.

The behavior of $\mathrm{RMSE}_{O}$ is monotone in $m$: as $m$ increases, the weight of the regularization shifts toward the missing entries, while the observed entries move closer to a perfect fit of the empirical covariance. Consequently, $\mathrm{RMSE}_{O}$ decreases with $m$. The value of $m$ that minimizes $\mathrm{RMSE}_{\mathrm{tot}}$ coincides with the elbow of the $\mathrm{RMSE}_{O}$ curve, i.e., the largest $m$ before $\mathrm{RMSE}_{O}$ undergoes a sustained decrease. A similar pattern is observed under the nonstationary covariance model (Appendix~\ref{appendix:simulation-study}, Figure~\ref{fig:param-m-nonstat}), and for different values of $\alpha$ (see Appendix~\ref{appendix:simulation-study}, Figures~\ref{fig:stat-paramsvarying} and~\ref{fig:nonstat-paramsvarying}, right panels).
These findings indicate that selecting $m$ at the elbow of the $\mathrm{RMSE}_{O}$ curve offers a practical criterion for approximating the minimizer of $\mathrm{RMSE}_{\mathrm{tot}}$, ensuring good overall estimation performance without requiring ground-truth information.


\begin{figure}[ht]
    \centering
    \includegraphics[width=0.9\linewidth]{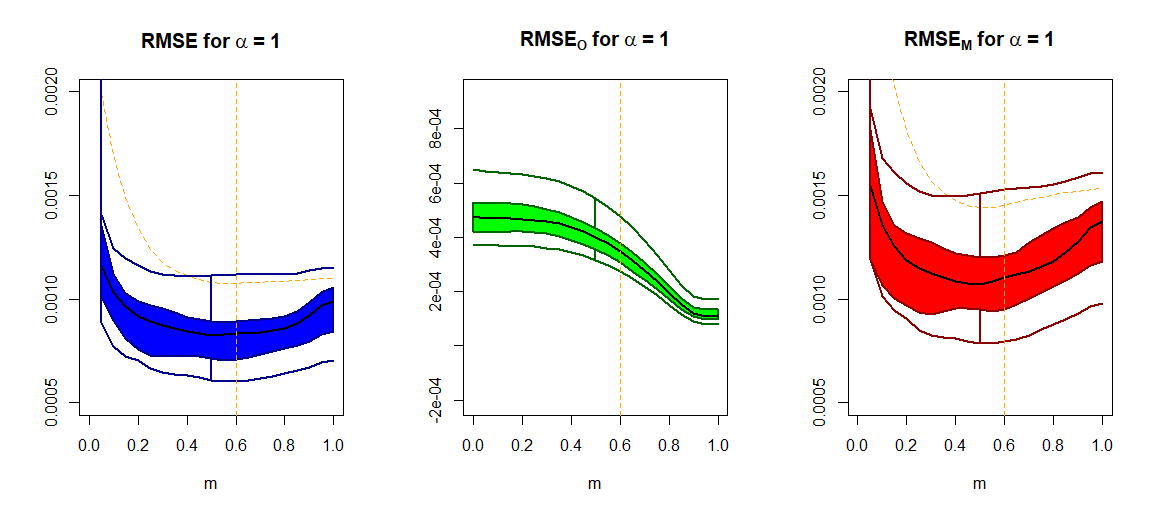}    
    \caption{Functional boxplots of the curves of RMSE generated for a grid of values of $m \in [0,1]$ over 30 simulations, with $\alpha = 1$ in the stationary scenario. Each functional boxplot is related to each type of error: $\mathrm{RMSE}_{\mathrm{tot}}$ (blue), RMSE$_O$ (green), RMSE$_M$ (red). The black curve indicates the median error over all replicates. In all three plots, the dashed vertical line is in correspondence the minimum of the mean of $\mathrm{RMSE}_{\mathrm{tot}}$ over all replicates. The vertical line in the plot of $\mathrm{RMSE}_{O}$ indicates the identified criterion for selecting $\alpha$. The dashed orange curves are outliers of the MC sample, identified by the functional boxplots.}
    \label{fig:param-m-stat}
\end{figure}


\paragraph{Visualization of the resulting estimate}
Figure~\ref{fig:preliminary-visualization} revisits the example introduced in Figure~\ref{fig:reconstruction-state-of-the-art} and displays the covariance estimate obtained with our method and the optimal choice of parameters according to the criteria outlined in this section. Unlike the results produced by state-of-the-art methods, our approach successfully recovers the covariance structure underlying the functional data in the fragmented regime.

\begin{figure}[ht]
    \centering
    \includegraphics[width=0.8\linewidth]{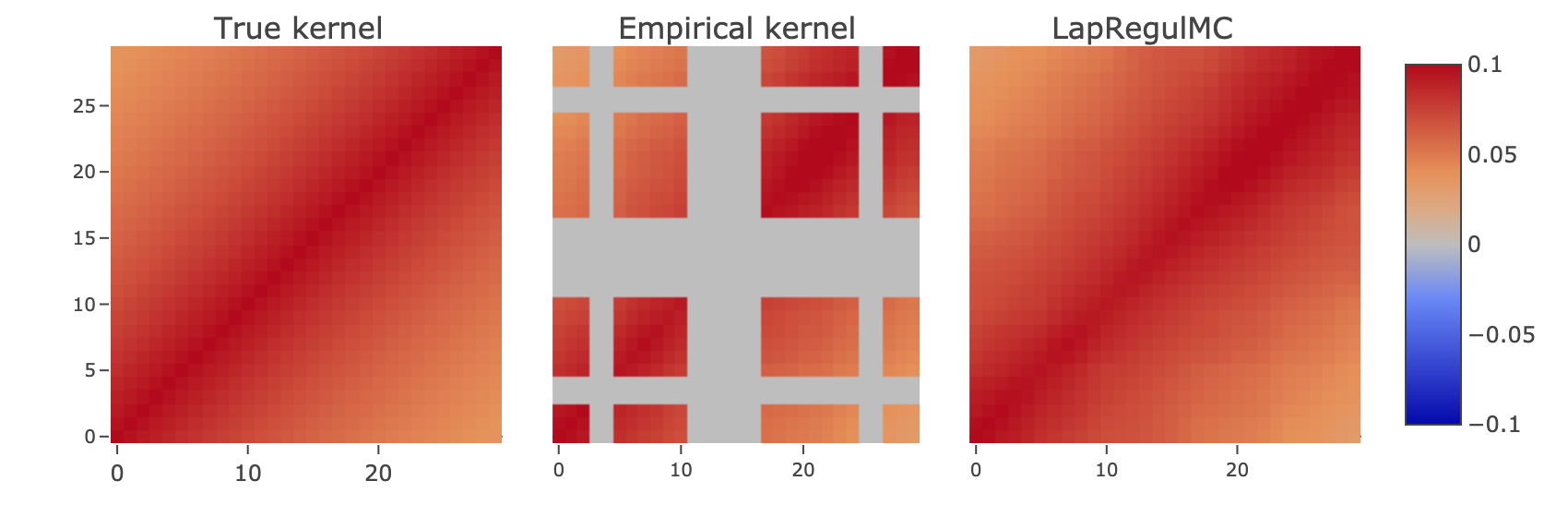}
    \caption{Example of covariance kernel estimation with state-of-the-art methods for partially observed functional data. From the left: true covariance kernel; empirical covariance kernel following~\eqref{eq:kernel_estimator_fragmented_regime}; covariance estimation with our Laplacian regularization approach ({\em LapRegulMC}). Other details are as in Figure~\ref{fig:reconstruction-state-of-the-art}.}
    \label{fig:preliminary-visualization}
\end{figure}


\section{Case study}
\label{sec:case-study}
We apply the methodology presented in Section~\ref{sec:method} to SBAS-processed SAR images of the Phlegraean Fields (Italy), an area prone to seismic and bradyseismic events. The parameters $\alpha$ and $m$ are selected following the data-driven criteria identified in Section~\ref{sec:hyperparam-sel}. After detailing how our framework generalizes to functional data with a two-dimensional domain, we estimate the covariance of the data and show its use for data reconstruction and scenario simulation.

\subsection{SBAS-Processed SAR data}
\label{sec:dinsar-data}
Synthetic Aperture Radar (SAR) is a remote sensing technology capable of acquiring high-resolution images regardless of lighting or weather conditions. Each pixel in a SAR image is associated to a complex number representing the reflected signal from the Earth’s surface. Differential SAR Interferometry (DInSAR) is a well-established microwave remote sensing technique that elaborates SAR data to measure ground deformation with centimeter to millimeter accuracy~\citep{gabriel1989mapping}. The DInSAR rationale is based on the use of the phase difference between two SAR images acquired over the same area at different times – i.e. the interferogram – to retrieve surface displacements occurred within the observation interval. The increasing availability of SAR data collected regularly over time made it possible to use the temporal evolution of interferograms to track the evolution of displacement phenomena \citep{ferretti2001permanent,berardino2002new}, rather than limiting the analysis to deformations occurring in a single time window. Multi-temporal DInSAR algorithms properly combine the information available from a set of independent interferograms to estimate the deformation time series of an observed area. Among all multi-temporal DInSAR techniques, the SBAS methodology \citep{berardino2002new} takes advantage of sequences of SAR data to generate accurate ground displacement time series \citep[][e.g.]{casu2006quantitative,manunta2019parallel} also at different spatial scales~\citep{bonano2012long}. The {\em temporal coherence} is an index in $[0,1]$ that quantifies the reliability of the retrieved time series at each pixel. In common practice, pixels with temporal coherence below 0.8 are considered to not provide trustworthy ground displacement measurements and are treated as missing.

\subsection{Preprocessing}
Our dataset consists of $n=391$ georeferenced SBAS-processed images measuring ground displacement in the Phlegraean Fields. After coherence thresholding, missingness occurs at the same pixel locations across all images. The SAR images have been acquired by the Copernicus Sentinel-1 constellation from 24/03/2015 to 22/02/2023 along descending orbits (Track 22). As illustrated in Figure~\ref{fig:coherence}, coherence is particularly low over water bodies, where SAR measurements are completely unreliable. Therefore, we restrict the analysis to a reliable subwindow of $101 \times 101$ pixels (red square in Figure~\ref{fig:coherence}), each covering $80\times80$ square meters.  

Prior to estimating the covariance, we remove temporal autocorrelation from the displacement series at each pixel. Specifically, after checking for stationarity, we fit moving average processes to each time series and retain the residuals. This preprocessing step is thoroughly described in Appendix~\ref{appendix:preprocessing}. The resulting residual maps, one per acquisition date, are treated as independent realizations of the underlying spatial process. This enables us to estimate the covariance kernel from $n=391$ independent samples of the process.

\begin{figure}[ht]
    \centering
    \includegraphics[width=1\linewidth]{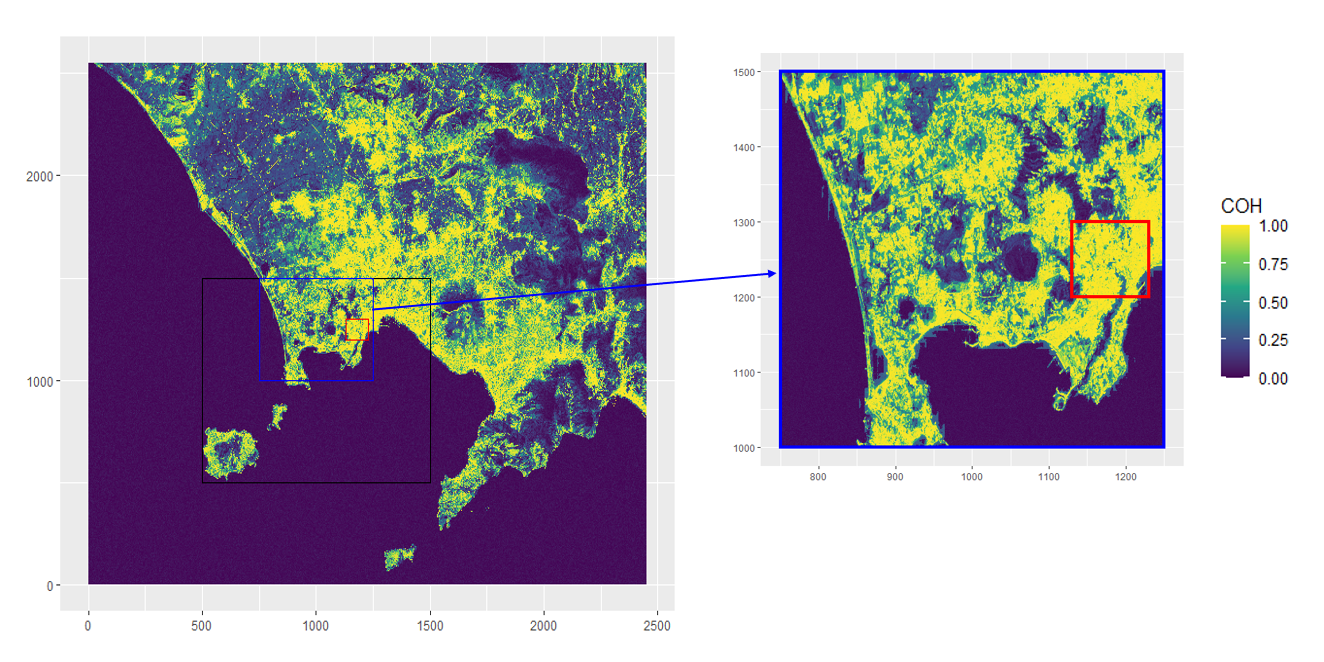}
    \caption{Temporal coherence values per pixel in the Phlegraean Fields (Italy). Pixels with coherence below 0.8 are treated as missing.}
    \label{fig:coherence}
\end{figure}

\subsection{Covariance estimation in two dimensions}
\label{sec:cov-rec-2D}
Each sample is defined on a two-dimensional grid, requiring regularization in both vertical and horizontal directions. To apply our Laplacian framework, we flatten each $101\times101$ grid into a one-dimensional vector by column-wise unrolling. This preserves vertical adjacency and, through indexing, also maintains horizontal neighborhood relationships. The covariance tensor between spatial locations can thus be represented as a $10201\times10201$ matrix with 41.25\% missing entries due to filtering based on the coherence values.  

To manage the high dimensionality and the nonstationary covariance structure, we initialize $\gamma$ as described in Section~\ref{sec:optimization-matrix-form} and impose a low-rank constraint during estimation. In accordance with the criteria prescribed in Section~\ref{sec:hyperparam-sel}, the parameter $\alpha$ is selected at the elbow point of the $\mathrm{RMSE}_O$ curve, yielding $\alpha_{\text{opt}}=1$, while $m$ is selected at the onset of a sharp decrease in $\mathrm{RMSE}_O$, yielding $m_{\text{opt}}=1$ (see Appendix~\ref{appendix:case-study-hyp-sel}, Figure~\ref{fig:error-alpha-m-cs}). Note that these values are consistent with the optimal settings found in Section~\ref{sec:simul-study} in the anisotropic and nonstationary scenario. Supporting plots are provided in Appendix~\ref{appendix:case-study-hyp-sel}.

Figure~\ref{fig:rec-cov-303} shows an illustrative estimation for three consecutive columns (dimension $303\times303$). The method is shown to successfully restore smooth spatial continuity both within and across columns, demonstrating that the unrolling preserves two-dimensional adjacency.

\begin{figure}[ht]
    \centering
    \includegraphics[width=0.8\linewidth]{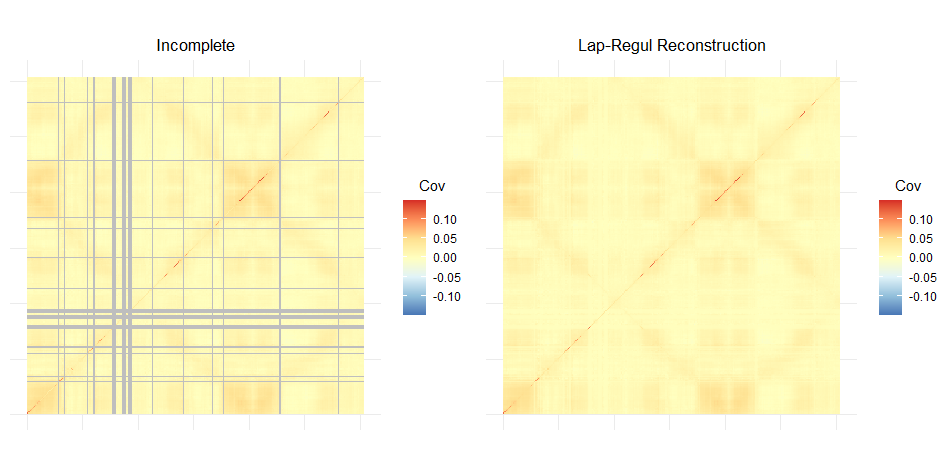}
    \caption{Illustrative covariance estimation for three consecutive columns of the $101\times101$ grid.}
    \label{fig:rec-cov-303}
\end{figure}

\paragraph{Comparison with other methods.}
As in the synthetic experiments, we benchmark our method against the Low-rank Matrix Completion~\citep{descary2019recovering} and a parametric geostatistical Matérn model {(\em Geostationary)}. Figure~\ref{fig:rec-cov-per-pix} displays, for each method, the estimated covariance maps for three reference pixels, which are marked by black dots in the plots. For pixels 5000 and 50 (top and bottom rows), which are fully observed across all images, the task is primarily to reconstruct their covariance with the missing pixels. Pixel 49 (central row) presents a more challenging task, as it is entirely unobserved in the dataset and its covariances with all other pixels is completely unobserved.

The results show that the {\em Geostationary} method performs well at capturing local spatial dependencies, particularly the covariance with nearby pixels—even, when those neighbors are unobserved. However, it fails to capture more complex spatial relationships, such as those driven by ground displacement, which do not strictly follow geographic proximity. On the other hand, the {\em Low-rankMC} method better captures these broader displacement patterns, but struggles to reconstruct the covariance structure of pixels that were themselves unobserved, like pixel 49. Our Laplacian Regularized Matrix Completion method appears to combine the strengths of both. It is capable of reconstructing fine scale spatial structure while also recovering broader, nonlocal dependencies driven by latent geophysical effects. This is particularly evident in the case of pixel 49, where our method successfully reconstructs the entire covariance structure despite the lack of any direct observations for that pixel. These results illustrate the robustness and flexibility of our approach in handling both partially and fully missing regions of the covariance matrix.

\begin{figure}[ht]
    \centering
    \includegraphics[width=0.9\linewidth]{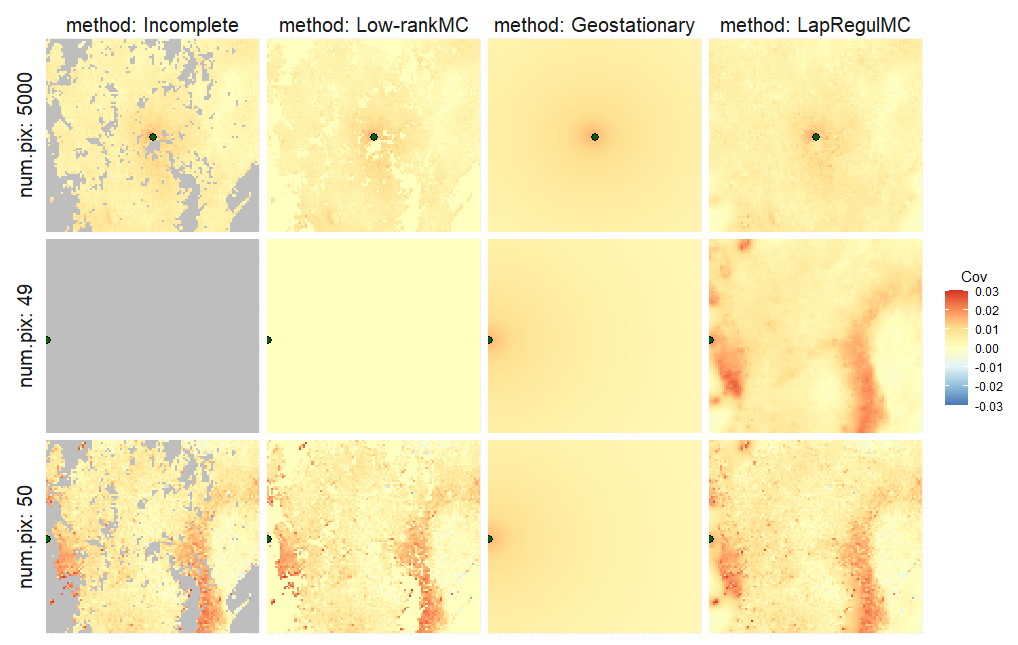}
    \caption{Comparison of per-pixel covariance estimates. Each row is referred to a reference pixel (indicated by a larger dot), and each column shows the covariance map produced by a different method: the initial incomplete empirical covariance, the Low-rank Matrix Completion method from \cite{descary2019recovering} {\em (Low-rankMC)}, a parametric geostatistical approach {\em (Geostationary)}, and the Laplacian Regularized Matrix Completion method {\em (LapRegulMC)} proposed in this work. In each heatmap, the value in a given pixel represents its estimated covariance with the reference pixel.}
    \label{fig:rec-cov-per-pix}
\end{figure}

\subsection{Applications of the estimated covariance}
\label{sec:case-study-2}
\paragraph{Images reconstruction.}

The estimated covariance can be employed to reconstruct the partially observed images following the reconstruction method for partially observed functional data of~\cite{kraus2015components}. This method aims at reconstructing missing parts of functional data via the best linear predictor, which depends on the cross- and auto-covariance operators. To improve the numerical stability of the resulting reconstruction, we apply a Ridge regularization with $\beta=0.001$~\citep{kraus2015components}. This leads to a regularized expression for reconstructing each partially observed image based on its observed part, the sample mean and the estimated covariance. Figure~\ref{fig:map-reconstr} presents the ground displacement reconstructions obtained using three possible covariance estimation strategies, namely the {\em Low-rankMC}, the {\em Geostationary} and the {\em LapRegulMC} methods. Image reconstruction is done for three acquisitions, namely 20 November 2016, 28 September 2017, and 1 August 2018. The figure illustrates how the covariance estimator affects the quality of the reconstructed displacement maps when observations are incomplete.
The {\em Low-rankMC} yields reconstructions that fail to recover meaningful spatial patterns in large missing areas. This behavior stems from the structure of the estimated covariance, for which the cross-covariance between observed and unobserved regions is effectively zero. As a consequence, the best linear predictor cannot propagate localized deformation patterns from the observed pixels to the missing ones, and the reconstruction reverts to the mean deformation field, especially in large gaps where no direct observations are available.
The {\em Geostationary} estimator produces smooth reconstructions, but at the cost of introducing overly diffuse spatial patterns in the regions originally corresponding to large areas of missing information.
In contrast, the proposed {\em LapRegulMC} method yields reconstructions that exhibit a smooth transition between observed and missing regions, with spatial patterns that adapt flexibly to the structure of each individual temporal sample.

\begin{figure}[ht]
    \centering
    \includegraphics[width=1\linewidth]{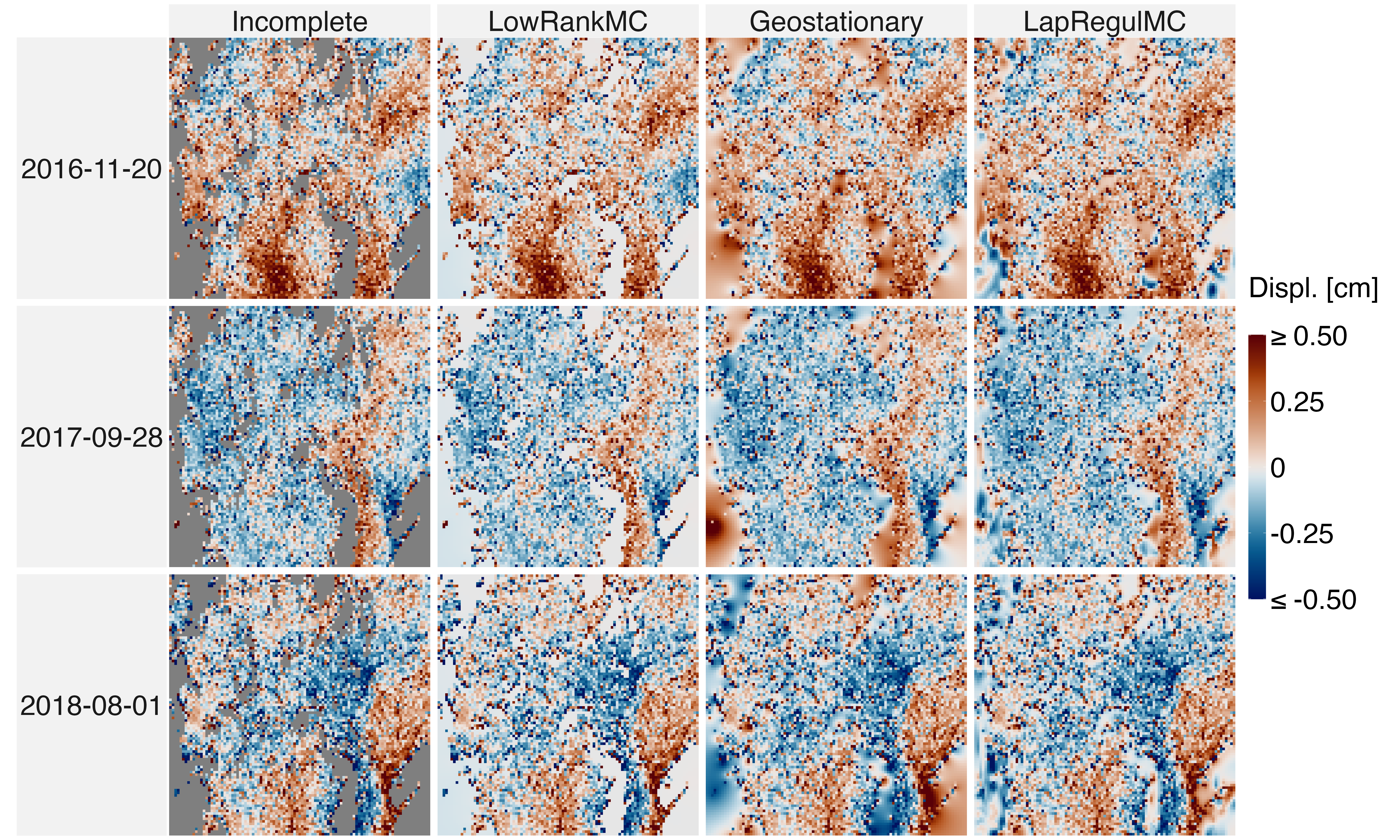}
    \caption{Incomplete and reconstructed ground displacement in images relative to an area of the Phlegraean Fields (Italy), on 20 November 2016, 28 September 2017, 1 August 2018. The reconstruction is done considering three possible estimation methods for the covariance: {\em Low-rankMC}, {\em Geostationary}, and {\em LapRegulMC}. A Ridge regularization with $\beta=0.001$ is performed in order to improve the numerical stability of the reconstruction.}
    \label{fig:map-reconstr}
\end{figure}

\paragraph{Scenario simulation.}
An accurate estimation of the covariance structure is central to scenario simulation, as it encodes the spatial dependencies that simulated fields must reproduce. Classical geostatistical simulation methods -- including exact approaches such as spectral simulation and LU decomposition, and approximate algorithms such as sequential Gaussian simulation and turning bands -- are all built around the availability of an accurate covariance, which guarantees that the generated realizations are statistically consistent with the target process \citep[e.g.,][]{davis1987production, dietrich1993fast, chiles2012geostatistics}. In applied settings, the credibility of scenario simulations -- and their value for uncertainty quantification and risk analysis -- hinges on how well the estimated covariance approximates the true one. This dependence is particularly evident in geophysical hazard modeling, where covariance‐based simulation has been used to merge InSAR and GNSS deformation rates~\citep{parizzi2020covariance}, or to generate scenarios of ground motion intensity measures ~\citep{abbasnejadfard2021investigating}.

Figure~\ref{fig:scenario_simulation} illustrates two simulated displacement scenarios generated from Gaussian processes whose covariance is estimated using the three competing approaches. The mean of the Gaussian process is the sample mean with smooth interpolations in the missing regions. Consistently with the behavior already observed in the reconstruction experiment, the simulations obtained with {\em Low-rankMC} exhibit poor spatial structure in the missing regions, reflecting the inability of the method to retrieve cross-dependence in the estimated covariance. The {\em Geostationary} estimator, on the other hand, produces overly smooth and spatially homogeneous fields, which tend to neglect localized deformation features. In contrast, simulations based on the proposed {\em LapRegulMC} covariance display richer and more coherent spatial variability, with deformation patterns that are spatially structured while remaining locally adaptive. This suggests that the regularized covariance estimate better captures the heterogeneous dependence structure of the displacement field, leading to more realistic simulated scenarios across the entire domain.

\begin{figure}[htbp]
    \centering
    \includegraphics[width=0.9\linewidth]{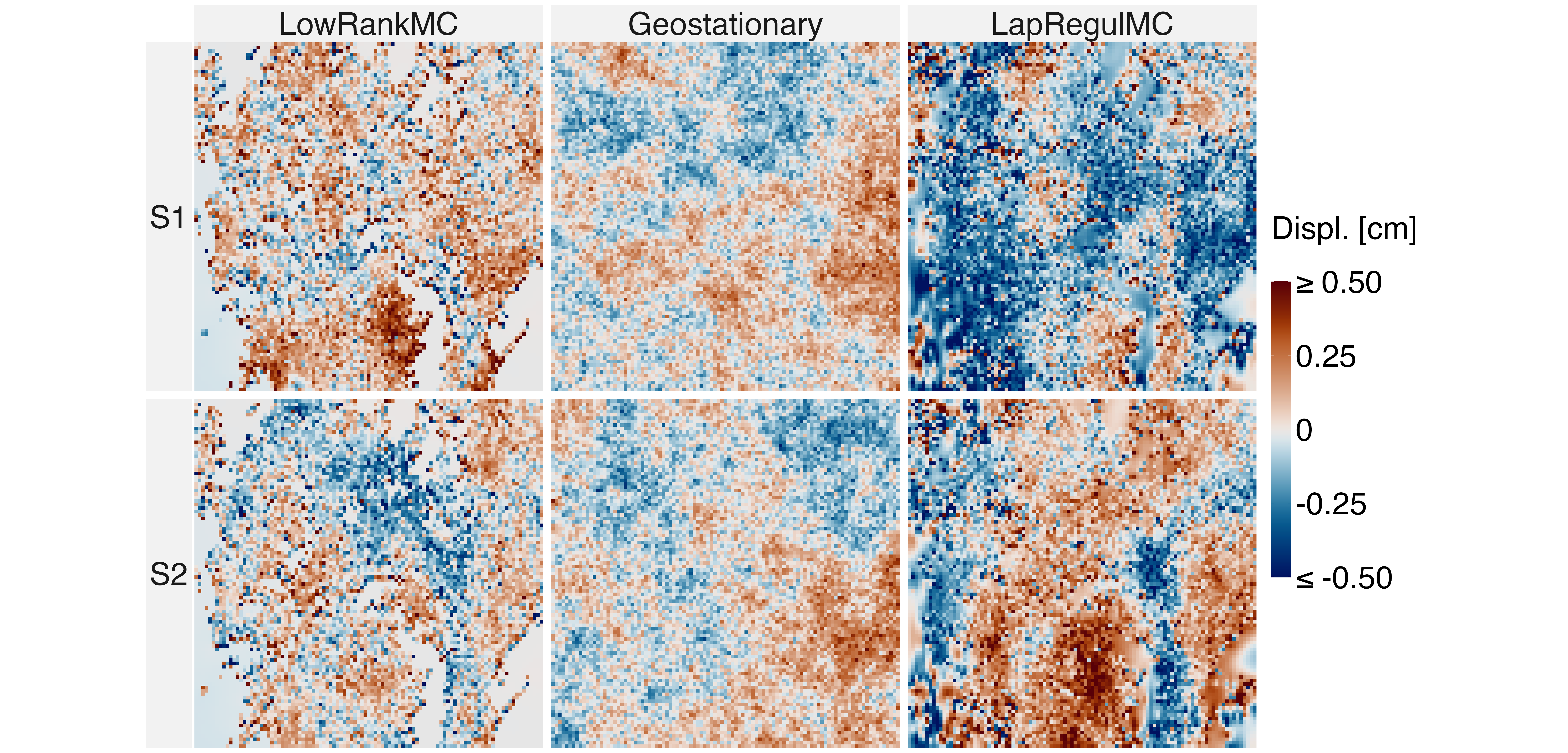}
    \caption{Simulated scenarios generated by sampling from a Gaussian process with second-order structure estimated using the {\em LapRegulMC} covariance estimation method.}
    \label{fig:scenario_simulation}
\end{figure}

\section{Discussion}
\label{sec:conclusions}
We have proposed a novel nonparametric approach to covariance estimation in the {\em fragmented regime}. By formulating the problem as matrix completion with Laplacian regularization, our method extends the covariance kernel smoothly into unobserved regions while preserving flexibility where data are available. The estimator is computationally tractable, it guarantees positive semidefiniteness, and it is straightforward to implement in practice. The proposed method is fully data-driven in that it identifies criteria for selecting the tuning parameters $\alpha$ and $m$ based on quantities that can be computed from the data. By leveraging the elbow points of the observable error in the covariance estimation as proxies for the minimizers of the total error, we obtain a practical parameter choice strategy that avoids computationally intensive cross-validation or resampling schemes.

Extensive simulations spanning stationary, nonstationary, isotropic, and anisotropic covariance structures demonstrate that {\em LapRegulMC} achieves consistently low estimation error. In contrast, low-rank matrix completion methods fail outside the serrated-domain regime, while the parametric geostatistical estimator deteriorates rapidly when its modeling assumptions are not met. These experiments highlight the flexibility and robustness of our approach across a broad spectrum of dependence structures.

We applied the methodology to ground displacement second-order analysis and reconstruction in the Phlegraean Fields. The estimated covariance exhibits smooth transitions from observed into missing regions, capturing spatial patterns that remain unrecovered by the alternative estimation methods considered in this work. Such estimation enables downstream tasks that depend on second-order structure, including ground displacement reconstruction in the missing regions and scenario simulation. These analyses further underscore the value and potential of the method for environmental hazard monitoring.

It should also be mentioned the recent work of~\citet{mbaka2025estimating}, which was developed concurrently with this work and also employs a Laplacian-type penalty for covariance estimation. While their framework shares a superficial similarity with ours, the focus and scope are different: their method targets measurement error correction and denoising, whereas our contribution addresses the challenge of reconstructing large unobserved regions in the fragmented regime. In this sense, the two approaches should be seen as complementary rather than overlapping.

While our work does not yet provide formal consistency guarantees, our simulation studies with increasing sample size suggest that the estimator behaves as expected and stabilizes toward the underlying covariance structure. Further theoretical investigation will be a natural next step, aimed at establishing identifiability conditions in the fragmented regime and at rigorously characterizing the asymptotic properties of the method.
In particular, if the unobserved regions were governed by processes with fundamentally different second-order structures than those observed at the boundaries, smooth extrapolation may not recover the true covariance.
Addressing these questions would complement the empirical evidence presented here and provide a deeper understanding of the statistical guarantees underlying our approach.

Future research may be devoted to addressing these limitations by studying identifiability conditions in the fragmented regime and developing theoretical results on the asymptotic behavior of the estimator. In scenarios where identifiability is not guaranteed, exogenous information may be incorporated in the estimation methodology to guide covariance estimation in the unobserved regions. One possibility in this direction would be to explore alternative forms of regularization which includes the additional information. This might help detect structural changes in the covariance kernel and reveal whether a different second-order process is active in the missing areas.



\bigskip

\noindent {\em Conflicts of interest:} None declared.

\section*{Acknowledgments}
The authors were partly supported by ACCORDO Attuativo ASI-POLIMI ``Attivit{\`a} di Ricerca e Innovazione" n.~2018-5-HH.0, collaboration agreement between the Italian Space Agency and Politecnico di Milano, and by the initiative ``Dipartimento di Eccellenza 2023–2027", MUR, Italy, Dipartimento di Matematica, Politecnico di Milano. They also acknowledge the financial support of IREA-CNR (Istituto per il Rilevamento Elettromagnetico dell'Ambiente del Consiglio Nazionale delle Ricerche) for funding a PhD grant in cooperation with Politecnico di Milano. GPT-5 was used for language improvement.

\section*{Data availability}
An open-source software implementation of the methods described in this paper, along  with the code needed to reproduce all numerical results from the simulations, are available online at \url{https://github.com/robertatroilo/Covariance_reconstruction.git}. The P-SBAS data that support the findings of this study are available from the corresponding author upon reasonable request.

\bibliographystyle{Chicago}

\bibliography{bibliography.bib}

\newpage

\appendix
\setcounter{theorem}{0}
\setcounter{equation}{0}
\setcounter{figure}{0}

\renewcommand{\thetheorem}{A\arabic{theorem}}
\renewcommand{\theproposition}{A\arabic{proposition}}
\renewcommand{\theequation}{A\arabic{equation}}
\renewcommand{\thefigure}{A\arabic{figure}}

\newcommand{\cA}{\mathcal{A}}
\newcommand{\cC}{\mathcal{C}}
\newcommand{\cD}{\mathcal{D}}
\newcommand{\cF}{\mathcal{F}}
\newcommand{\cI}{\mathcal{I}}
\newcommand{\hDelta}{\hat{\Delta}}

\section{Performance of state-of-the-art methods in the fragmented regime}
\label{appendix:state-of-the-art}

Figure~\ref{fig:exp2_alpha0.1_m1_K30_rank20_reduced} reports the reconstructions obtained with the methods of~\citet{descary2019recovering} and~\citet{delaigle2021estimating} when missing intervals are very narrow, corresponding to isolated missing locations in the discretized setting. The method of~\citet{descary2019recovering} assigns zero covariance to all pairs involving a missing location. The method of~\citet{delaigle2021estimating}, while not directly setting covariances to zero, already oversmooths them toward zero in the missing regions of the empirical kernel.
Figures~\ref{fig:exp3_alpha0.1_m1_K30_rank20_reduced} and~\ref{fig:exp4_alpha0.1_m1_K30_rank20_reduced} display the reconstructions provided by the two methods as the width of the missing intervals increases. The results clearly show that the tendency to set covariances to zero becomes increasingly severe for both methods as the fragmented regime worsens.

\begin{figure}[ht]
    \centering
    \includegraphics[width=1\linewidth]{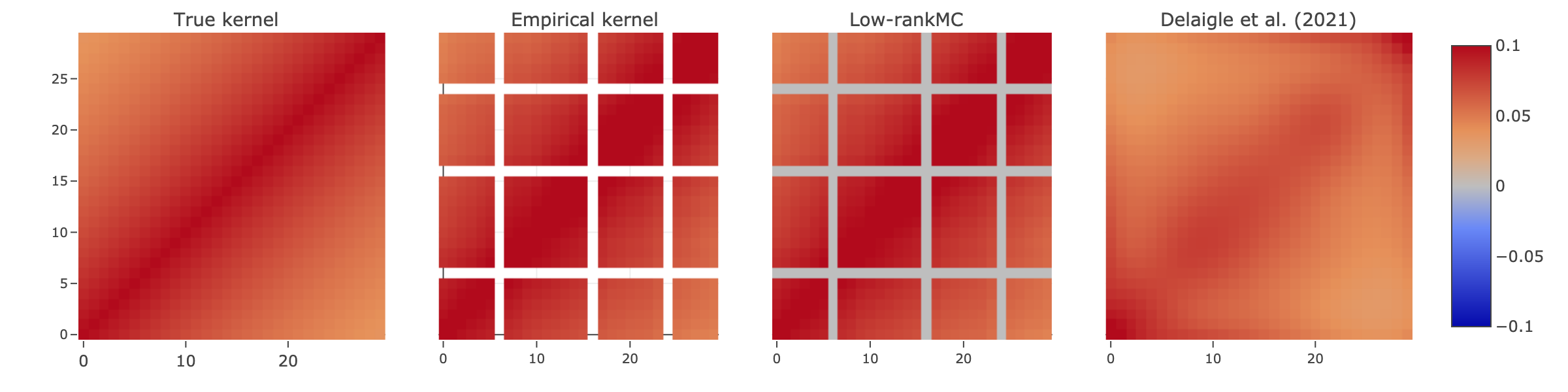}
    \caption{Example of covariance kernel reconstruction with state-of-the-art methods for partially observed functional data in the fragmented regime. The fragmented regime is simulated by selecting $J_M = \{7,17,25\}$ as the set of missing indexes along the grid (see Section~\ref{sec:data-sim}). Other details are as in Figure~\ref{fig:reconstruction-state-of-the-art}.}
    \label{fig:exp2_alpha0.1_m1_K30_rank20_reduced}
\end{figure}

\begin{figure}[ht]
    \centering
    \includegraphics[width=1\linewidth]{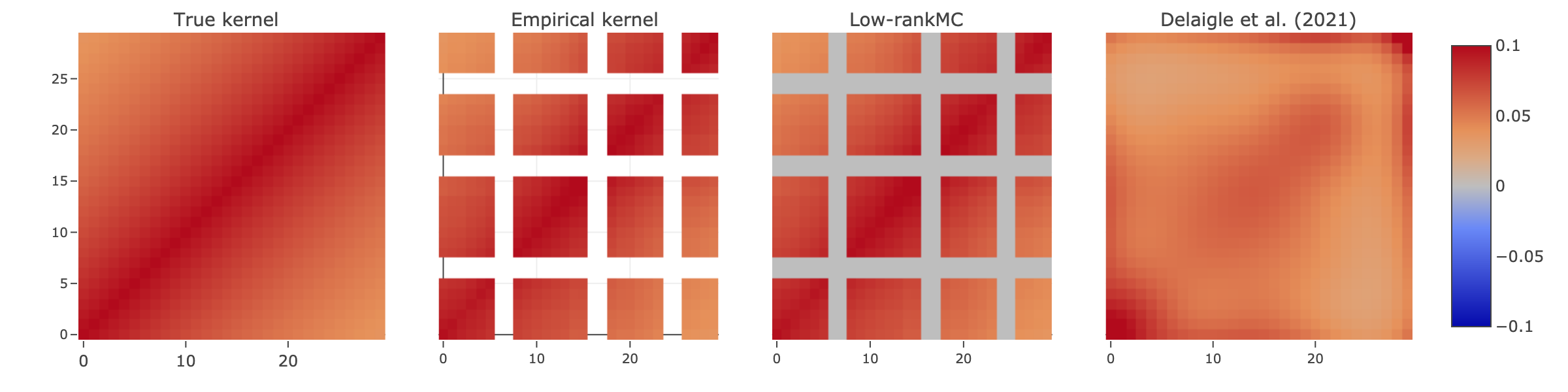}
    \caption{Example of covariance kernel reconstruction with state-of-the-art methods for partially observed functional data in the fragmented regime. The fragmented regime is simulated by selecting $J_M = \{7:8,17:18,25:26\}$ as the set of missing indexes along the grid (see Section~\ref{sec:data-sim}). Other details are as in Figure~\ref{fig:reconstruction-state-of-the-art}.}
    \label{fig:exp3_alpha0.1_m1_K30_rank20_reduced}
\end{figure}

\begin{figure}[ht]
    \centering
    \includegraphics[width=1\linewidth]{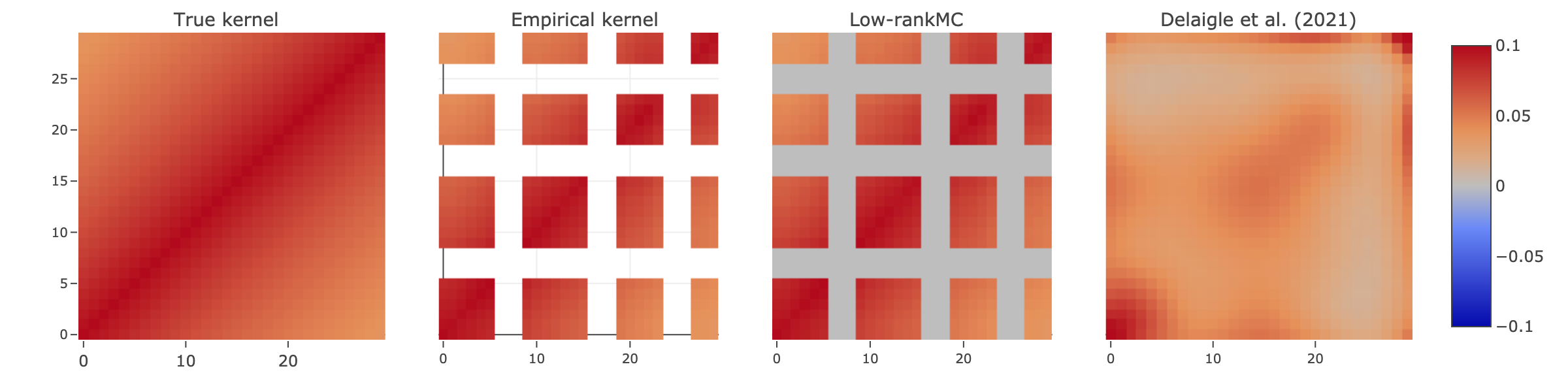}
    \caption{Example of covariance kernel reconstruction with state-of-the-art methods for partially observed functional data in the fragmented regime. The fragmented regime is simulated by selecting $J_M = \{7:9,17:19,25:27\}$ as the set of missing indexes along the grid (see Section~\ref{sec:data-sim}). Other details are as in Figure~\ref{fig:reconstruction-state-of-the-art}.}
    \label{fig:exp4_alpha0.1_m1_K30_rank20_reduced}
\end{figure}


\section{Mathematical proof}
\label{appendix:proof}

We first state and prove two simpler results.

\begin{proposition}
\label{prop:wellposed-simplified}
The optimization problem
\begin{equation}
    \underset{\theta \in \mathbb{S}_{+}^K}{\text{min}} \left\{ 
        \left\| P^O \circ (\hat{R}^K - \theta) \right\|_F^2 
        + \alpha
            \left\| \left( L_1\theta + \theta L_1 \right) \right\|_F^2
    \right\},
    \label{eq:min-pb-theta-simplified}
\end{equation}
admits a unique minimizer $\theta^\star\in\mathbb{S}_+^K$.
\end{proposition}

\begin{proof}[Proof of Proposition~\ref{prop:wellposed-simplified}]
The proof relies on the well-known result that, if a set $C\subset\mathbb{R}^n$ is nonempty, closed, and convex, and
\begin{align*}
    f(x)=\frac{1}{2} x^T Qx + b^T x + c,
\end{align*}
with $Q$ symmetric positive definite, $b\in \mathbb{R}^n$ and $c \in \mathbb{R}$, then $f$ admits a unique minimizer on $C$ (e.g., \citealt{boyd2004convex})

Since the space $\mathbb{S}^+_K$ is nonempty, closed and convex, the proof is completed by expressing the objective function~\eqref{eq:min-pb-theta} as a quadratic functional and by showing that $Q$ is positive definite.

\noindent \textit{Notation.} In the following, we denote with $\mathrm{diag}$ the operation of extracting the sole diagonal of a square matrix, and with $\mathrm{Diag}$ the operation of diagonalizing a vector. That is, if $x \in \mathbb{R}^d$, then $\mathrm{Diag}(x)$ is a $d$-dimensional diagonal matrix and the following equality holds: $\mathrm{diag}(\mathrm{Diag}(x)) = x$.

\textit{Expressing the objective as a quadratic functional.}
Let $x=\mathrm{vec}(\theta)\in\mathbb{R}^{K^2}$ and $y=\mathrm{vec}(\hat R^K)$.
Let $D=\mathrm{Diag}(\mathrm{vec}(P^O))$ and let
\begin{align*}
    M \;=\; L_1\otimes I \;+\; I \otimes L_1.
\end{align*}

Using $\mathrm{vec}(P^O \circ (\hat R^K-\theta))= D \; \mathrm{vec}(\hat R^K-\theta)$ and 
$\mathrm{vec}(L_1\theta+\theta L_1)= M\,\mathrm{vec}(\theta)$, 
the objective reads
\begin{align*}
    f(x) = \|P^O\circ(\hat R^K-\theta)\|_F^2 + \alpha\|L_1\theta+\theta L_1\|_F^2
    = \|D(y-x)\|_2^2 \;+\; \alpha\|Mx\|_2^2 ,
\end{align*}
where $\| \cdot \|_2$ denotes the Eucledian norm.
Then,
\begin{align*}
    f(x) &= (y-x)^\top D (y-x) + \alpha x^\top M^\top M x \\
    &= x^\top (D+\alpha M^\top M) x \;-\; 2 y^\top D x \;+\; y^\top D y,
\end{align*}
where we used $D^2=D$ since $D$ is a diagonal matrix with values 0 and 1.

Hence,
\begin{align*}
    f(x) \;=\; \tfrac12 x^\top Q x \;+\; b^\top x \;+\; c,
\end{align*}
with
\begin{align*}
    Q \;:=\; 2\big(D+\alpha M^\top M\big),\quad b \;:=\; -2 D y , \quad c:=y^\top D y.
\end{align*}
Matrix $Q$ is clearly symmetric (recall that $D$ is diagonal).
Thus showing $Q$ is positive definite will imply that $f$ is strictly convex and coercive.

\textit{Positive definiteness of $Q$.}
For any $x$,
\begin{align*}
    x^\top Q x \;=\; 2\|D x\|_2^2 \;+\; 2 \alpha\|Mx\|_2^2.
\end{align*}
Thus $Q$ is positive definite iff $\ker(D) \cap \ker(M)=\{0\}$. We characterize the two kernels separately, and show that their intersection is empty.

\begin{enumerate}
    \item Since $D$ is a diagonal $0$–$1$ projector onto the observed entries, its nullspace is
    \begin{align*}
        \ker(D)=\{x\in\mathbb{R}^{K^2}:\ x_{pq}=0 \text{ for every observed pair }(p,q)\in\Omega_O\}.
    \end{align*}
    Equivalently, $\ker(D)$ are exactly the matrices that vanish on all observed pairs in the set $\Omega_O$.

    \item It is known that the discrete Neumann Laplacian $L_1$ is symmetric positive semidefinite with
    \begin{align*}
        \ker(L_1)=\mathrm{span}\{u_1\},\quad u_1\propto\mathbf{1},
    \end{align*}
    and admits an orthonormal eigenbasis $L_1=U\Lambda U^\top$ with
    $0=\lambda_1<\lambda_2\le\cdots\le\lambda_K$.
    For the Kronecker \emph{sum}
    \begin{align*}
        M \;=\; L_1\otimes I \;+\; I\otimes L_1,
    \end{align*}
    the spectral property of Kronecker sums gives eigenpairs
    \begin{align*}
        \{\,(\lambda_i+\lambda_j,\ u_i\otimes u_j)\,\}_{i,j=1}^K.
    \end{align*}
    Hence
    \begin{align*}
        \ker(M)=\mathrm{span}\{u_1\otimes u_1\} =\mathrm{span}\{\mathrm{vec}(u_1u_1^\top)\} =\mathrm{span}\{\mathrm{vec}(\mathbf{1}\mathbf{1}^\top)\},
    \end{align*}
    i.e., the only matrices unpenalized by the smoothness term are constant rank-one matrices.
\end{enumerate}

Since the vectorization of any nonzero constant matrix does not belong to $\ker{D}$, it follows that $\ker(D)\cap\ker(M)=\{0\}$ and $Q$ is positive definite.
\end{proof}

\begin{proposition}
\label{prop:wellposed-masked}
The optimization problem
\begin{equation}
    \underset{\theta \in \mathbb{S}_{+}^K}{\text{min}} \left\{ 
        \left\| P^O \circ (\hat{R}^K - \theta) \right\|_F^2 
        + \alpha
            \left\| \bar{I} \circ \left( L_1\theta + \theta L_1 \right) \right\|_F^2
        + \alpha
            \left\| I \circ \left( L_1 (I \circ \theta) J  \right) \right\|_F^2
    \right\},
    \label{eq:min-pb-theta-masked}
\end{equation}
admits a unique minimizer $\theta^\star\in\mathbb{S}_+^K$.
\end{proposition}

\begin{proof}[Proof of Proposition~\ref{prop:wellposed-masked}]
We refer to the notation introduced in the proof of Proposition~\ref{prop:wellposed-simplified}.

\textit{Expressing the objective as a quadratic functional.}
Let $D_I = \mathrm{Diag}(\mathrm{vec}(I))$ and $D_{\bar{I}} = \mathrm{Diag}(\mathrm{vec}(\bar{I}))$. Define
\begin{align*}
    A_1 := D_{\bar{I}} \; M,
    \qquad
    A_2 := D_I (J \otimes L_1) D_I.
\end{align*}
Then the objective equals
\begin{align*}
    f(x)=\|D(y-x)\|_2^2+\alpha\|A_1 x\|_2^2+\alpha\|A_2 x\|_2^2.
\end{align*}
In vectorized form this is a quadratic $f(x)=\tfrac12 x^\top Qx+b^\top x+c$ with
\begin{align*}
    Q=2\big(D+\alpha A_1^\top A_1+\alpha A_2^\top A_2\big), \quad b \;:=\; -2 D y , \quad c:=y^\top D y.
\end{align*}
We now show that $Q$ is positive definite.

\textit{Positive definiteness of $Q$.}
It suffices to show that, within the PSD cone, the only matrix lying in both kernels of $D$ and $A_2$ is the zero matrix:
\begin{align*}
    \mathbb{S}_+^K \cap \ker D \cap \ker A_2 \;=\; \{0\}.
\end{align*}
Then the same holds when intersecting also with $\ker A_1$, since
\begin{align*}
    \mathbb{S}_+^K \cap \ker D \cap \ker A_2 \cap \ker A_1
\;\subseteq\;
\mathbb{S}_+^K \cap \ker D \cap \ker A_2
\;=\; \{0\}.
\end{align*}

\begin{enumerate}
    \item As in the proof of Proposition~\ref{prop:wellposed-simplified}, $\mathrm{ker}(D)$ are the matrices that vanish on all observed pairs in the set $\Omega_O$.
    \item Writing $d_{\theta}=\mathrm{diag}(\theta)\in\mathbb{R}^K$, note that
        \begin{align*}
            (I\circ\theta)=\mathrm{Diag}(d_{\theta}),\quad L_1 (I\circ\theta) J  = \big(L_1 d_{\theta}\big)\,\mathbf{1}^\top, \quad I\circ\big(L_1(I\circ\theta)J\big)=\mathrm{Diag}(L_1 d_{\theta}).
        \end{align*}
        Hence
        \begin{align*}
            \|A_2(\theta)\|_F^2=\|L_1 d_{\theta}\|_2^2, \qquad
            \ker(A_2)=\{\theta:\ \mathrm{diag}(\theta)\in\ker(L_1)\}
            = \{\theta:\ \mathrm{diag}(\theta)=c\,\mathbf{1}\},
        \end{align*}
        where we used again the fact $\ker(L_1)=\mathrm{span}\{\mathbf{1}\}$ for the discrete Neumann Laplacian.
\end{enumerate}

Following these observations, we are able to make the following argument.
Assume $\theta\in\mathbb{S}_+^K$ satisfies simultaneously
\begin{align*}
    P^O\circ \theta=0, \qquad A_2(\theta)=0.
\end{align*}
From $A_2(\theta)=0$ we get $\mathrm{diag}(\theta)=c\,\mathbf{1}$ for some $c\in\mathbb{R}$. Since at least one diagonal entry is observed (we have at least one systematically observed point in our setting, hence at least one entry in the diagonal of the empirical covariance matrix), the condition $P^O\circ\theta=0$ forces $c=0$.

It is a known result that a positive semidefinite matrix has a zero entry on its main diagonal if and only if the entire row and column to which that entry belongs is zero (see ~\citealp{horn2012matrix}, p. 434). Therefore, $\mathrm{diag}(\theta)=0$ implies $\theta=0$, and we can conclude that
\begin{align*}
    \mathbb{S}_+^K\cap\ker D\cap\ker A_2 = \mathbb{S}_+^K\cap\ker D\cap\ker A_2 \cap\ker A_1=\{0\}.
\end{align*}

\end{proof}

We are now ready to prove Proposition~\ref{prop:wellposed}.

\begin{proof}[Proof of Proposition~\ref{prop:wellposed}.]
We build on the same rationale followed in the proofs of Propositions~\ref{prop:wellposed-simplified} and ~\ref{prop:wellposed-masked}.

\textit{Expressing the objective as a quadratic functional.}
Recall that, from the proof of Proposition~\ref{prop:wellposed-masked},
\begin{align*}
    A_1 := D_{\bar{I}} \; M,
    \qquad
    A_2 := D_I (J \otimes L_1) D_I.
\end{align*}
For the definition of $M$ and $D_I$, we refer to Propositions~\ref{prop:wellposed-simplified} and ~\ref{prop:wellposed-masked}.

Let $D_{P^m} = \mathrm{Diag}(\mathrm{vec}(P^m))$. Then, define
\begin{align*}
    A'_1 := D_{P^m} \, A_1,
    \qquad
    A'_2 := D_{P^m} \, A_2.
\end{align*}
Then, it is possible to express the objective~\eqref{eq:min-pb-theta} as a quadratic function
\begin{align*}
    f(x)=\|D(y-x)\|_2^2+\alpha\|A'_1 x\|_2^2+\alpha\|A'_2 x\|_2^2.
\end{align*}
In vectorized form this is a quadratic $f(x)=\tfrac12 x^\top Qx+b^\top x+c$ with
\begin{align*}
    Q=2\big(D+ \alpha \, A_1'^\top A_1' + \alpha \, A_2'^\top A_2' \big), \quad b \;:=\; -2 D y , \quad c:=y^\top D y.
\end{align*}
We now show that $Q$ is positive definite.

\textit{Positive definiteness of $Q$.}
As in the proof of Proposition~\ref{prop:wellposed-masked}, it is sufficient to show that the intersection of the kernel of $D$ and the kernel of $A'_2$ is empty.

As in the proof of Proposition~\ref{prop:wellposed-simplified}, $\mathrm{ker}(D)$ are the matrices that vanish on all observed pairs in the set $\Omega_O$.

We now turn to the kernel of $A_2'$.
Let $d_{P^m} = \mathrm{diag}(P^m)$. Following similar steps as in Proposition~\ref{prop:wellposed-masked} one can show that
\begin{align*}
    \lVert P^m \circ I \circ (L_1 (I \circ \theta) J)\rVert_F^2 = \lVert d_{P^m} \circ L_1 d_{\theta} \rVert_2^2,
\end{align*}
where $d_{\theta} = \mathrm{diag} (\theta)$.

If $m \in (0,1)$, then the proof follows exactly the same argument as in the proof of Proposition~\ref{prop:wellposed-masked}.
If $m=1$, then a matrix $\theta$ belongs to $\mathrm{ker} \, (A_2')$ if the missing entries in its main diagonal, i.e. $i$ such that $d_{P^m,i} =[P^m]_{ii} = 1$, are constant and anchor with continuity to the observed entries in the main diagonal, i.e. $i'$ such that $d_{P^m,i'} = [P^m]_{i'i'} = 0$. Using again the argument in~\citet{horn2012matrix}, p. 434, this corresponds to a non-null positive semidefinite matrix $\theta$ only if $d_{\theta,i'} \neq 0$. Having $d_{\theta,i'} \neq 0$, with $(i',i')$ being an observed entry in $\theta$, implies that $\theta \notin \mathrm{ker} \, D$. This concludes the proof.

\end{proof}

\section{Simulation study}
\label{appendix:simulation-study}

\subsection{Comparison with existing methods}
\label{appendix:comparison-with-existing-methods}

\begin{figure}[ht]
    \centering
    \includegraphics[width=0.8\linewidth]{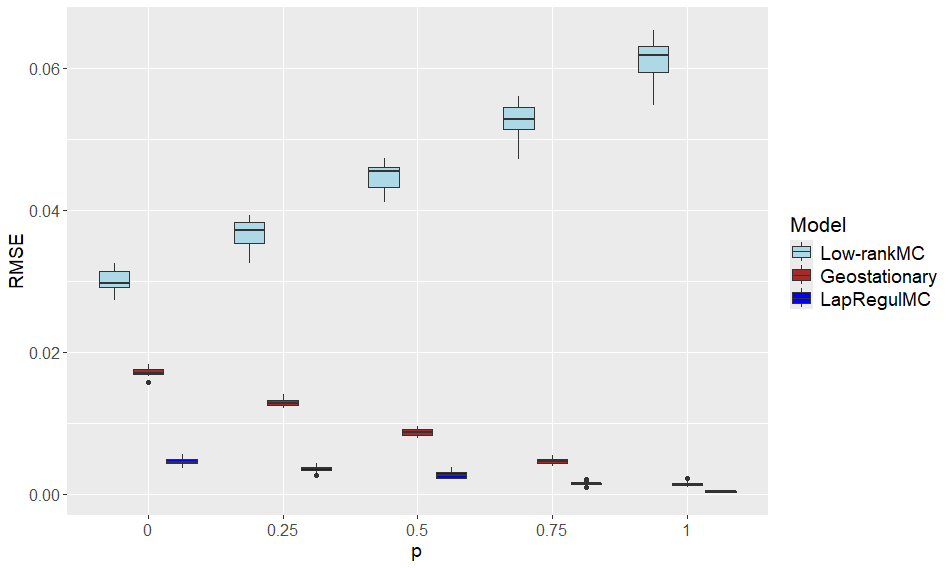}
    \caption{Comparison of the performance in terms of $\text{RMSE}_{\mathrm{tot}}(\hat{R}^K, \hat{\theta})$ for the three models for covariance reconstruction: the method proposed in \cite{descary2019recovering} (Low-rankMC), a parametric geostationary method (Geostationary) and our method (LapRegulMC).}
    \label{fig:methods-comparison-emp}
\end{figure}

\subsection{The impact of relying on finite sample size}
\label{appendix:increasing_nsamples}
In this section, we aim to assess the impact of selecting the hyperparameters based on the empirical covariance matrix, rather than the true covariance, on the performance of the proposed estimation routine.

We consider three distinct estimation errors: (i) {\em Empirical–empirical error (EE)} denotes the total error $\mathrm{RMSE}_{\mathrm{tot}}$ computed with respect to the empirical covariance matrix, when the hyperparameters are selected to minimize that same empirical error; (ii) {\em True–true error (TT)} denotes the total error $\mathrm{RMSE}_{\mathrm{tot}}$ computed with respect to the true covariance matrix, when the hyperparameters are selected to minimize the error with respect to the true covariance; (iii) {\em True–empirical error (TE)} denotes the total error $\mathrm{RMSE}_{\mathrm{tot}}$ computed with respect to the true covariance matrix, when the hyperparameters are selected as the minimizers of the empirical error.

Figure~\ref{fig:increasing_nsamples} displays boxplots of three errors obtained over 30 independent experiments, for increasing sample sizes $n \in \{10^2, 10^3, 10^4, 10^5\}$, under the stationary (left) and the nonstationary (right) covariance structures introduced in Section~\ref{sec:data-sim}. The pattern of partial observation is sampled randomly at each replication of the experiment.
Across both scenarios, the distributions of the errors become increasingly similar as $n$ grows. In particular, the {\em ET} and {\em TT} errors are similar across all $n$, indicating that selecting the hyperparameters by minimizing the empirical error yields performance comparable to the (infeasible) choice based on the true covariance.

\begin{figure}[ht]
    \centering
    \subfloat{%
        \includegraphics[width=0.49\linewidth]{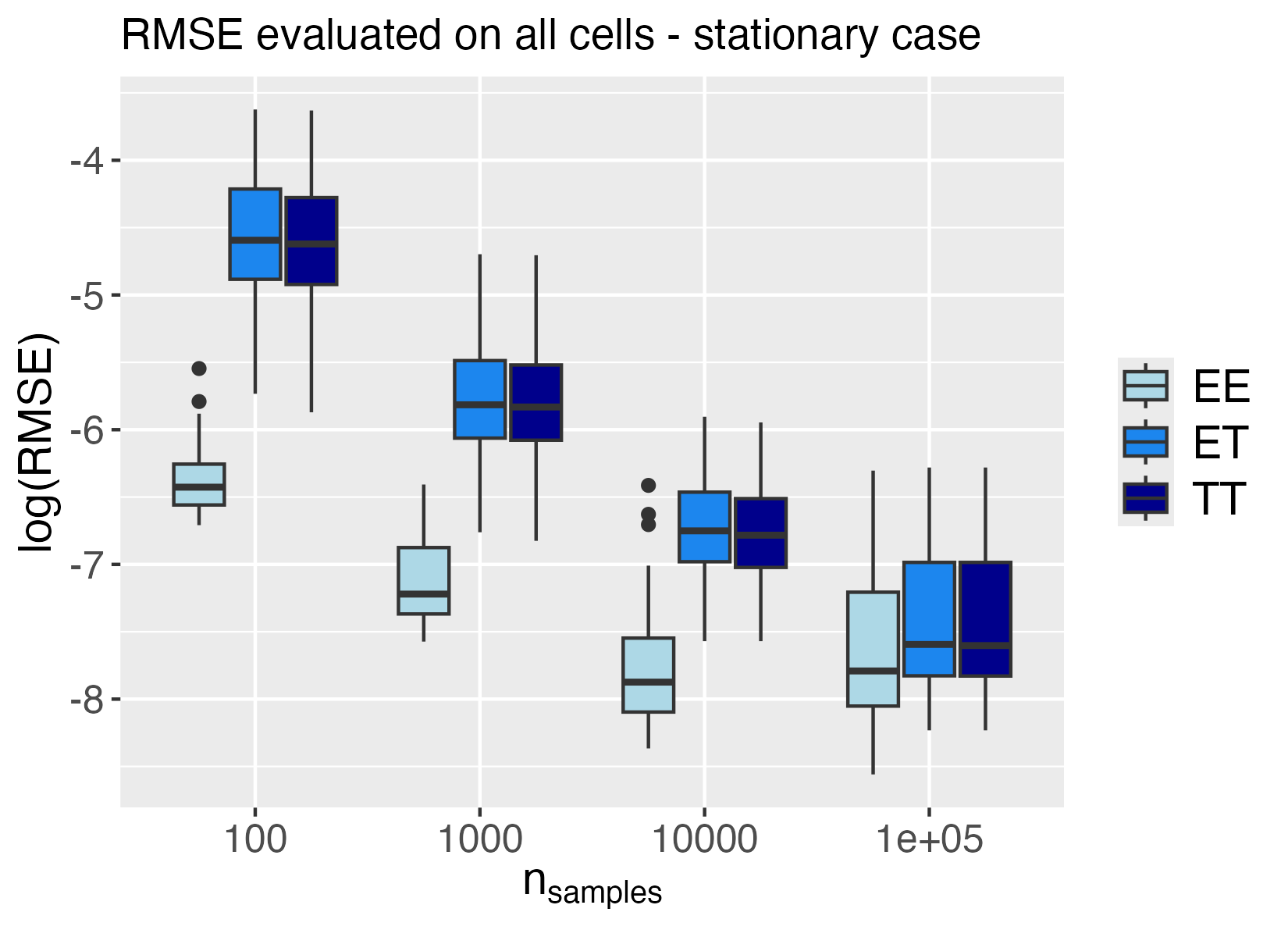}
    }
    \hfill
    \subfloat{%
        \includegraphics[width=0.49\linewidth]{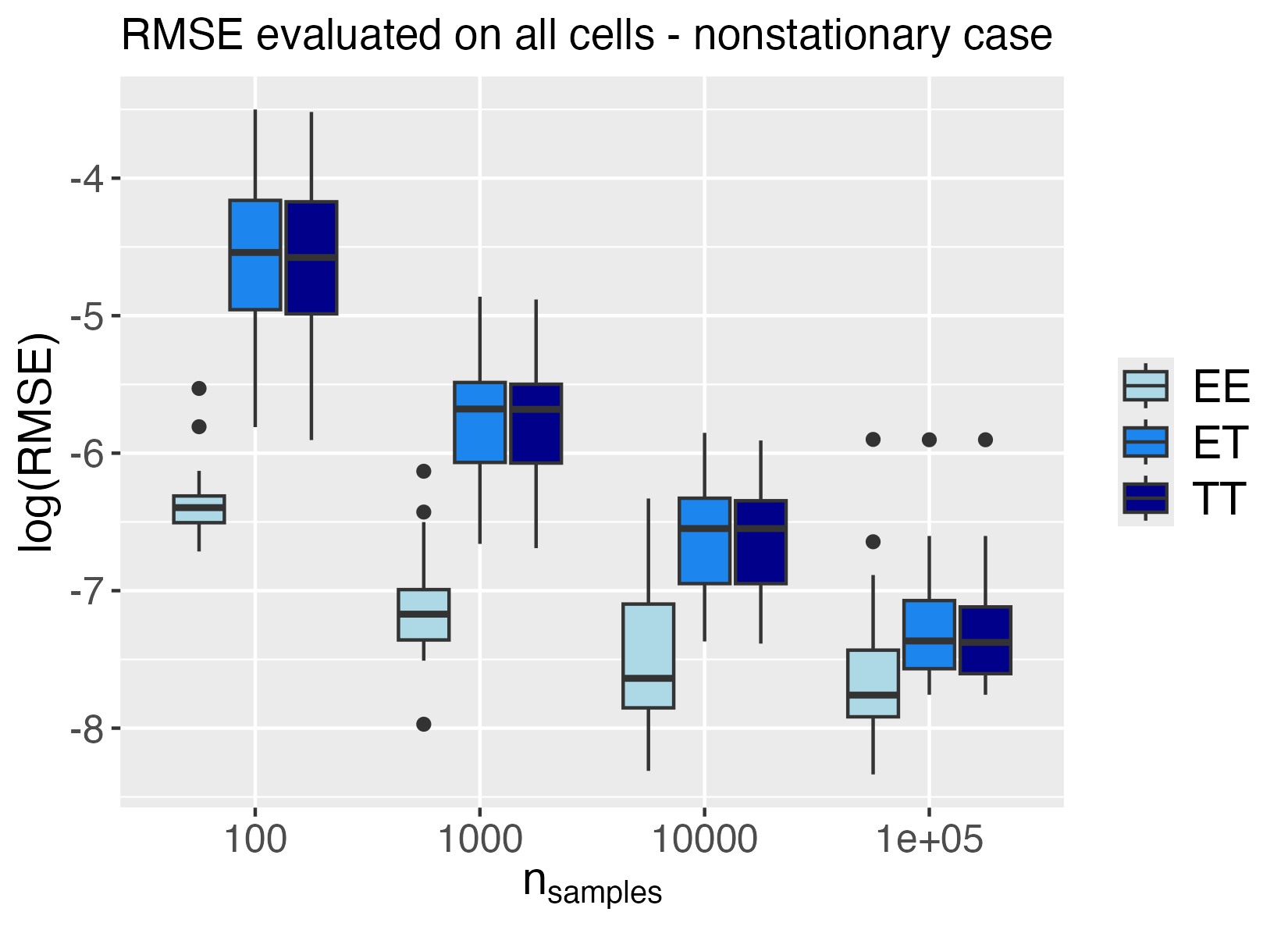}
    }
    \caption{Boxplots of three estimation errors over 30 independent experiments for increasing sample sizes. The considered covariance structures are the Matérn stationary covariance (left) and the nonstationary covariance (right) introduced in Section~\ref{sec:data-sim}. {\em EE} denotes the estimation error with respect to the empirical covariance using empirically optimized hyperparameters. {\em TT} denotes the estimation error with respect to the true covariance using hyperparameters optimized for the true covariance. {\em TE} denotes the estimation error with respect to the true covariance when hyperparameters are selected by minimizing the empirical error.}
    \label{fig:increasing_nsamples}
\end{figure}

To complement this analysis, Figure~\ref{fig:increasing_nsamples_optimal_stat} displays the distributions of the optimal $\alpha$ (left) and $m$ (right) obtained by minimizing the empirical error ({\em E}) and the true error ({\em T}) for the Matérn covariance model. The distributions of the two errors become increasingly similar as the sample size grows, showing that the hyperparameters identified by tuning using the error with respect to the empirical covariance converge to those obtained using the error with respect to the true covariance. An analogous behavior is observed for the nonstationary covariance model, as shown in Figure~\ref{fig:increasing_nsamples_optimal_nonstat}.

\begin{figure}[ht]
    \centering
    \subfloat{%
        \includegraphics[width=0.49\linewidth]{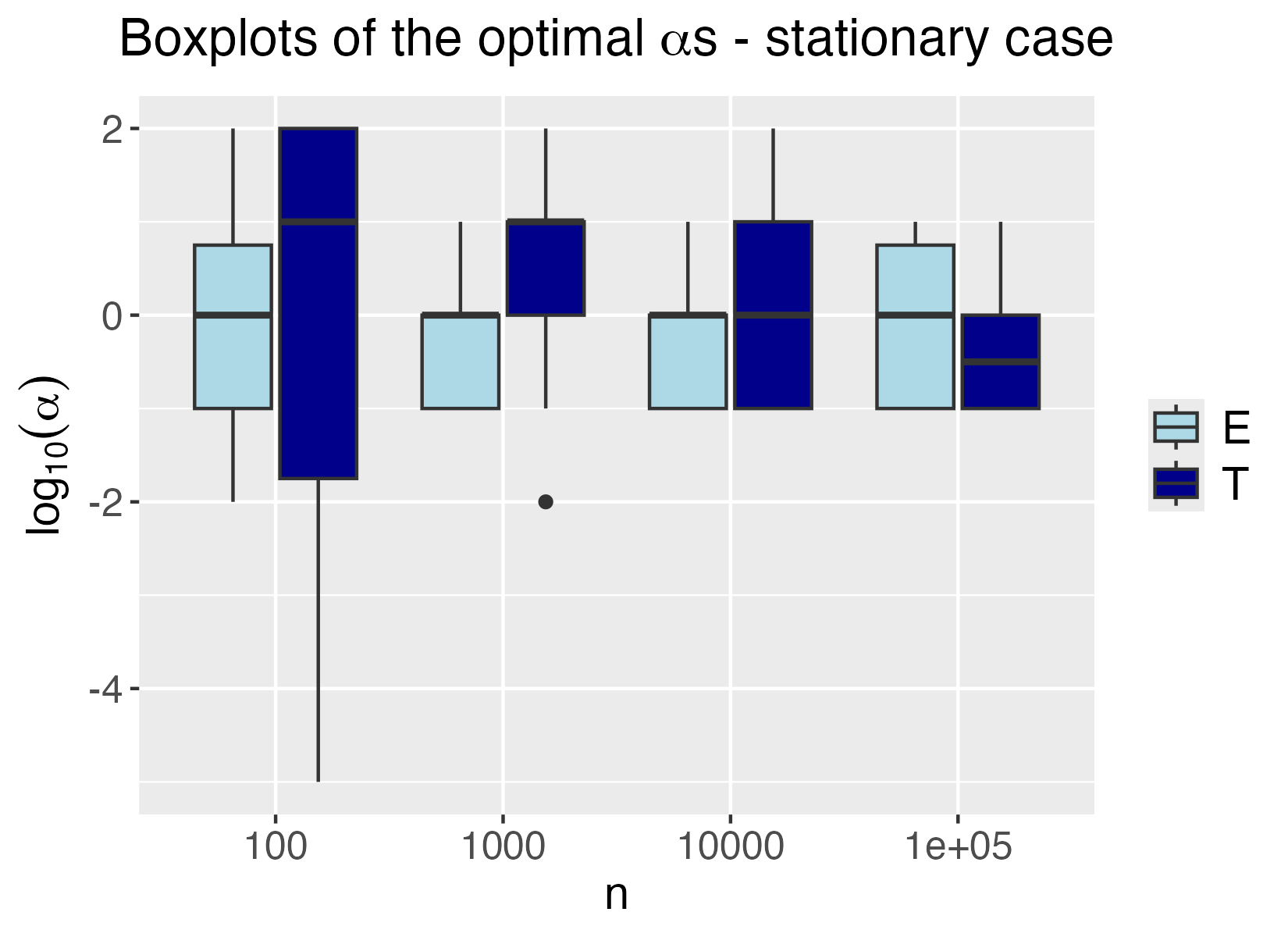}
    }
    \hfill
    \subfloat{%
        \includegraphics[width=0.49\linewidth]{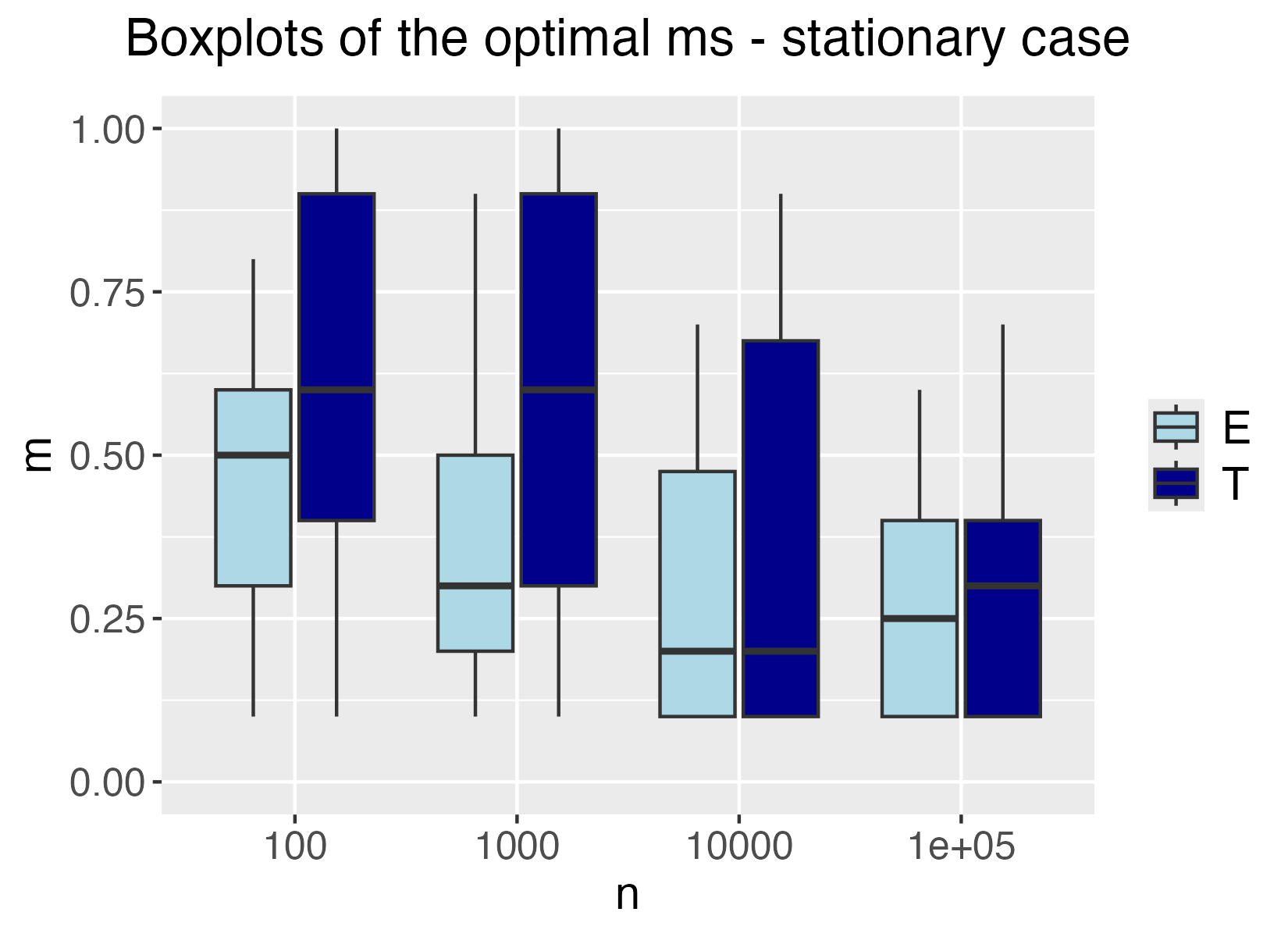}
    }
    \caption{Boxplots of the optimal parameters $\alpha$ (left) and $m$ (right) identified as the minimizers of the total error with respect to the empirical covariance ({\em E}) and the total error with respect to the true covariance ({\em T}). The considered covariance structure is the Matérn stationary covariance introduced in Section~\ref{sec:data-sim}.}
    \label{fig:increasing_nsamples_optimal_stat}
\end{figure}

\begin{figure}[ht]
    \centering
    \subfloat{%
        \includegraphics[width=0.49\linewidth]{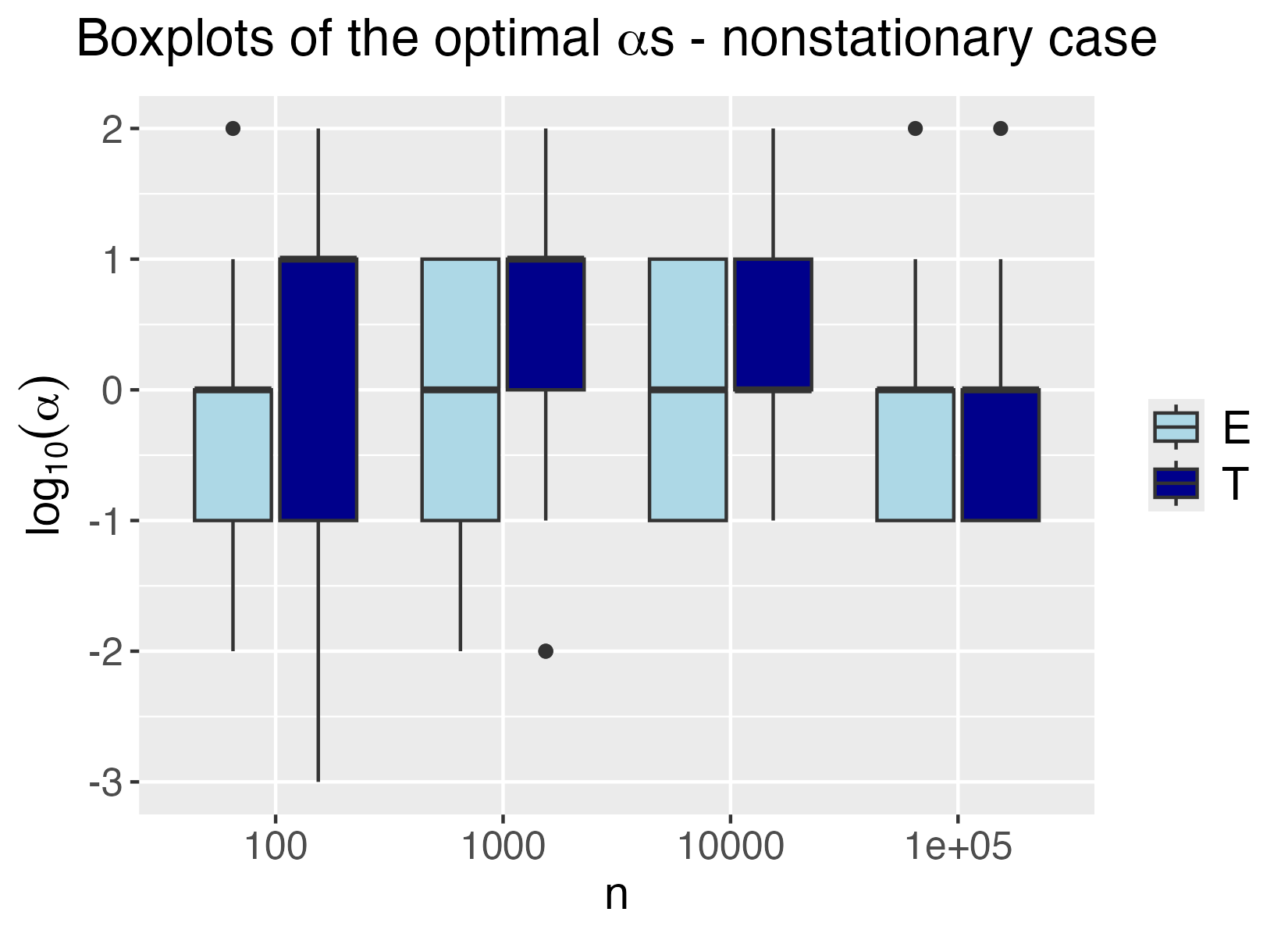}
    }
    \hfill
    \subfloat{%
        \includegraphics[width=0.49\linewidth]{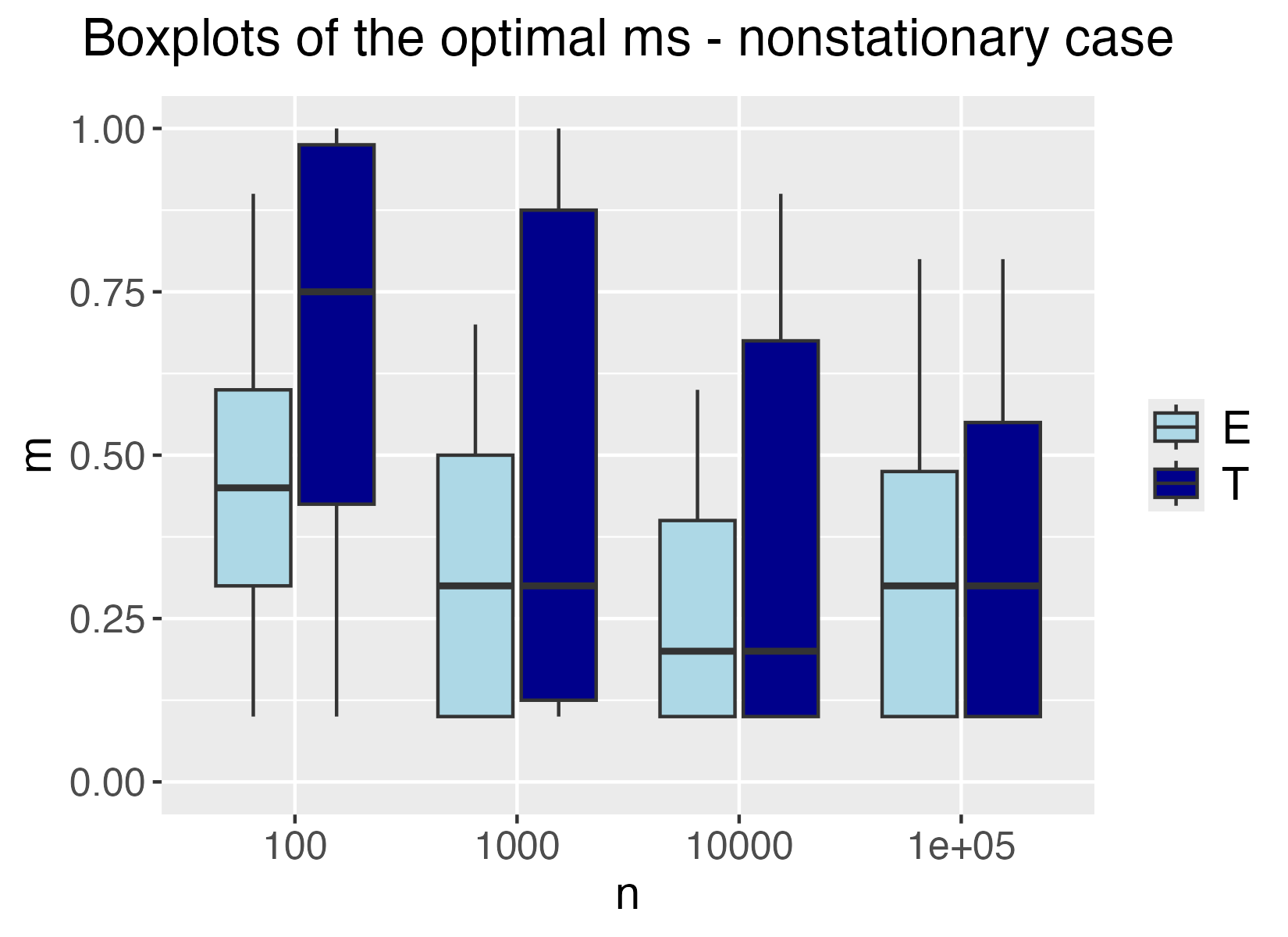}
    }
    \caption{Boxplots of the optimal parameters $\alpha$ (left) and $m$ (right) identified as the minimizers of the total error with respect to the empirical covariance ({\em E}) and the total error with respect to the true covariance ({\em T}). The considered covariance structure is the nonstationary covariance introduced in Section~\ref{sec:data-sim}.}
    \label{fig:increasing_nsamples_optimal_nonstat}
\end{figure}

Taken together, these results show that, as the sample size increases, the empirical covariance provides an increasingly accurate approximation for the true covariance for the purpose of hyperparameter selection. Consequently, minimizing the estimation error with respect to the empirical covariance -- the only error criterion available in practice -- yields asymptotically equivalent hyperparameter choices and estimation performance.
In Section~\ref{sec:hyperparam-sel}, we discuss how hyperparameters are selected in real applications, where the total empirical error cannot be evaluated and only the error on observed entries of the empirical covariance is available.

\subsection{Practical criteria for hyperparameters selection}
\label{appendix:simulation-study-hyp-sel}

\begin{figure}[ht]
    \centering
    \includegraphics[width=0.8\linewidth]{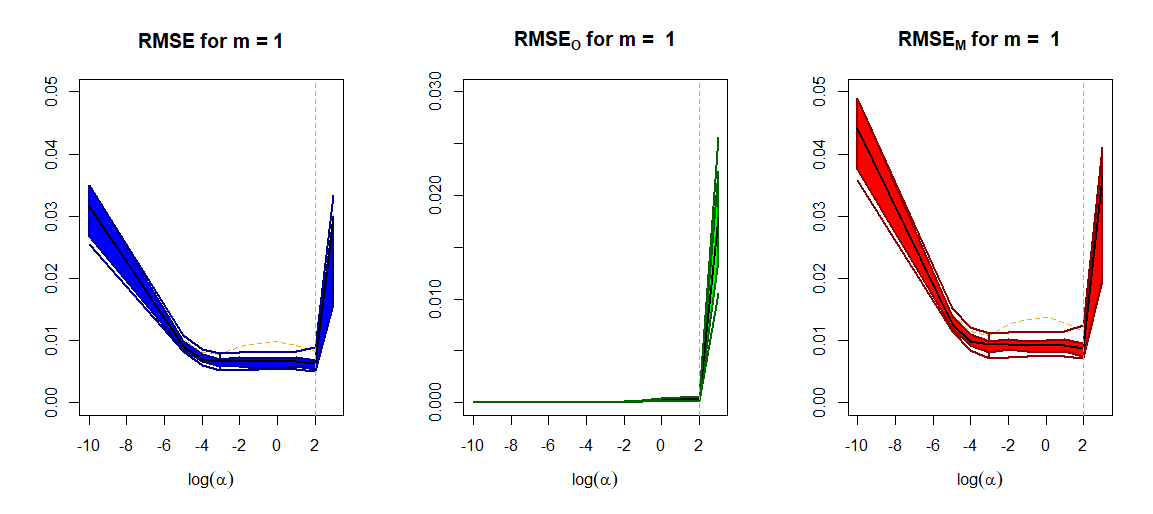}
    \caption{Functional boxplots of the curves of RMSE generated per each value of the logarithm of \(\alpha\) over 30 simulations, for \(m=1\) in the nonstationary scenario. Each functional boxplot is related to a different error: total RMSE (blue), RMSE\(_O\) (green), RMSE\(_M\) (red). The point \(log(\alpha)=-10\) on the x-axis serves as a fictitious representation graphically introduced to illustrate the value of RMSE for \(\alpha = 0\).}
    \label{fig:param-alpha-nonstat}
\end{figure}

\begin{figure}[ht]
    \centering
    \includegraphics[width=0.8\linewidth]{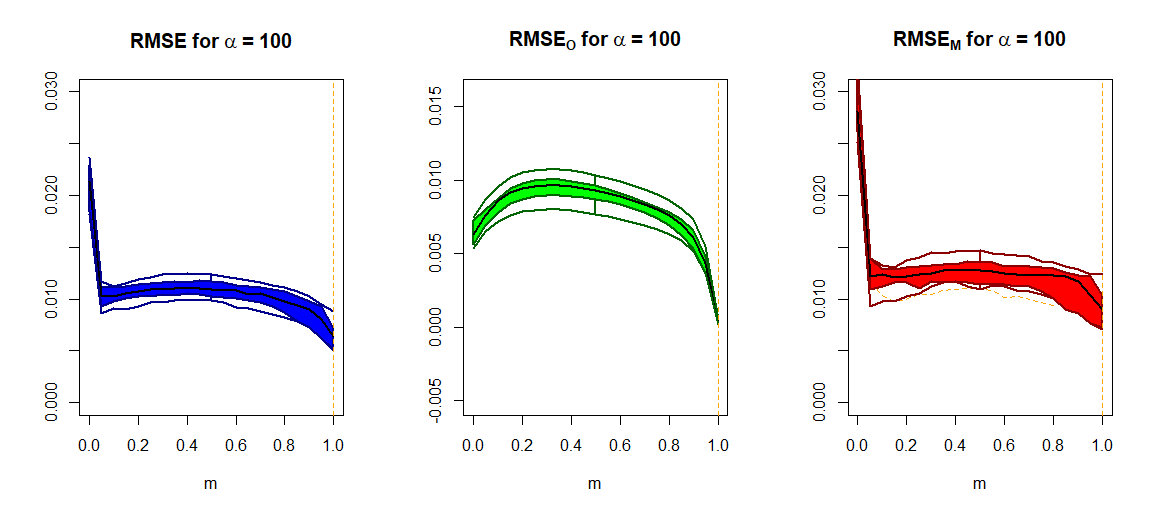}
    \caption{Functional boxplots of the curves of RMSE generated for a grid of values of $m \in [0,1]$ over 30 simulations, for \(\alpha=100\) in the nonstationary scenario. Each functional boxplot is related to each type of error: total RMSE (blue), RMSE\(_O\) (green), RMSE\(_M\) (red).}
    \label{fig:param-m-nonstat}
\end{figure}

\begin{figure}[ht]
    \centering
    \includegraphics[width=0.8\linewidth]{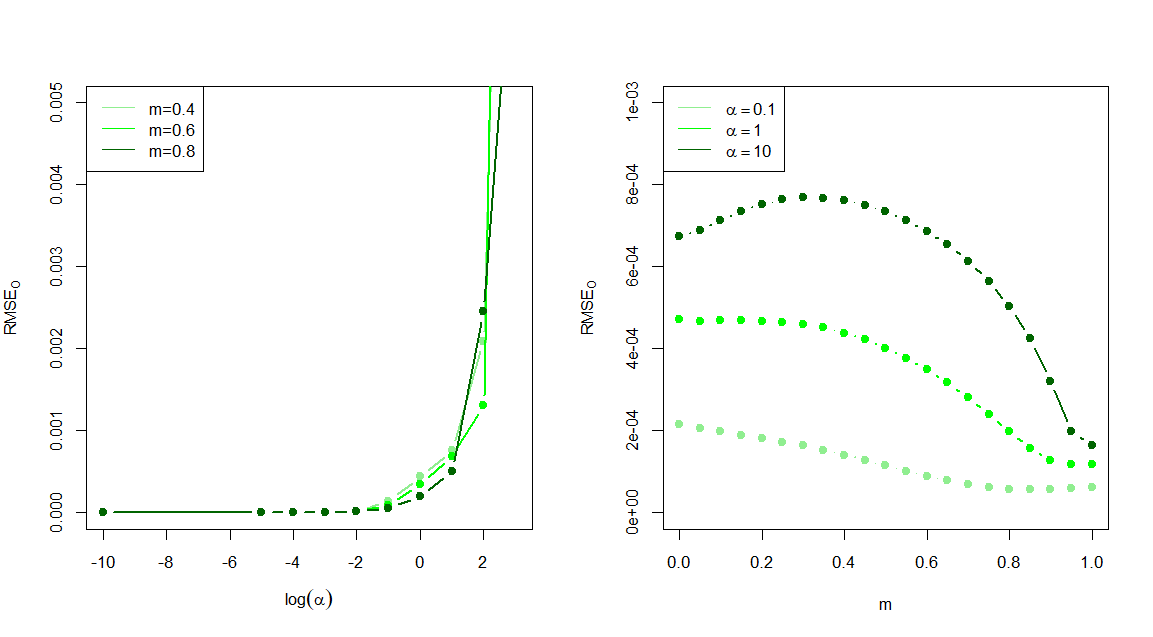}
    \caption{Curves of mean RMSE generated for each value of $\log \alpha$ over 30 simulations, for different values of $m$ (left) and of $\alpha$ (right) in the stationary scenario.}
    \label{fig:stat-paramsvarying}
\end{figure}

\begin{figure}[ht]
    \centering
    \includegraphics[width=0.8\linewidth]{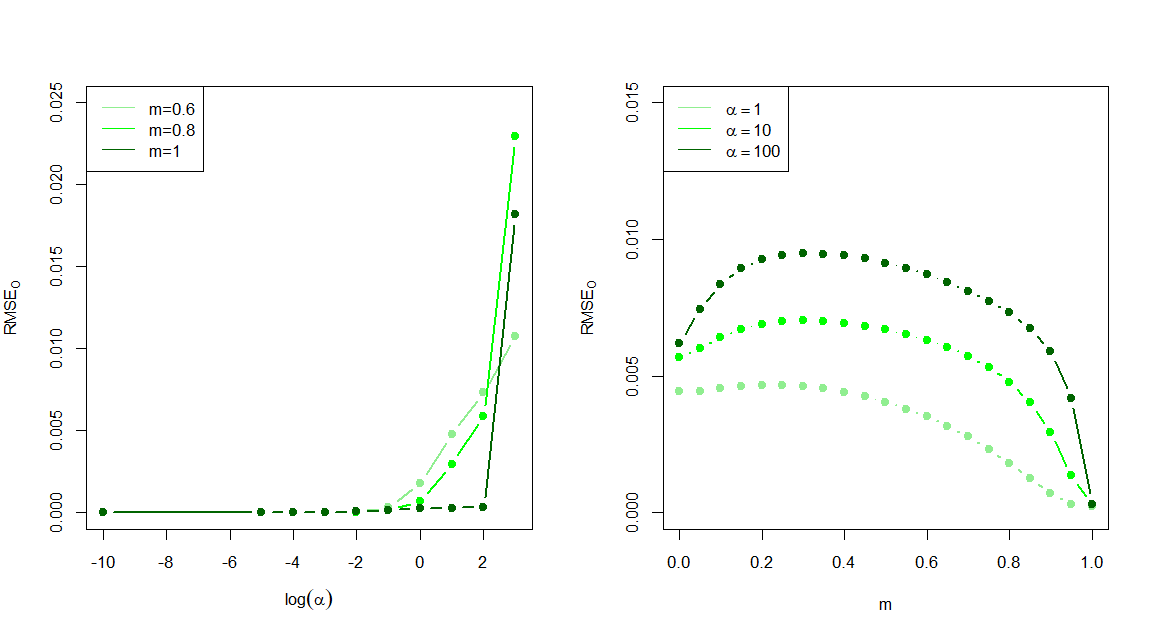}
    \caption{Curves of mean RMSE generated for each value of $\log \alpha$ over 30 simulations, for different values of $m$ (left) and of $\alpha$ (right) in the nonstationary scenario.}
    \label{fig:nonstat-paramsvarying}
\end{figure}


\section{Case study}
\label{appendix:case-study}

\subsection{Preprocessing}
\label{appendix:preprocessing}
The available data cover the period from 27/03/2015 to 22/02/2023. The spacing between consecutive SBAS acquisitions is not constant and gives rise to three distinct time intervals. From 27/03/2015 to 27/09/2016, observations are separated by 12 days; from 27/09/2016 to 30/12/2021, by 6 days; and from 30/12/2021 to 22/02/2023, again by 12 days. Since the subsequent analysis requires equally spaced observations to account for temporal dependence, the three time intervals are treated separately.

We present the temporal analysis on the second time interval, which comprises 316 out of 391 total observations, and a similar procedure was applied to the other time intervals. Autocorrelation is assessed pixel-wise using the Durbin–Watson statistic \citep{durbin1950testing}. Across all pixels, the statistic consistently lies between 2 and 4, providing clear evidence of negative autocorrelation (see Fig. \ref{fig:dw-test}).

\begin{figure}[ht]
    \centering
    \includegraphics[width=1\linewidth]{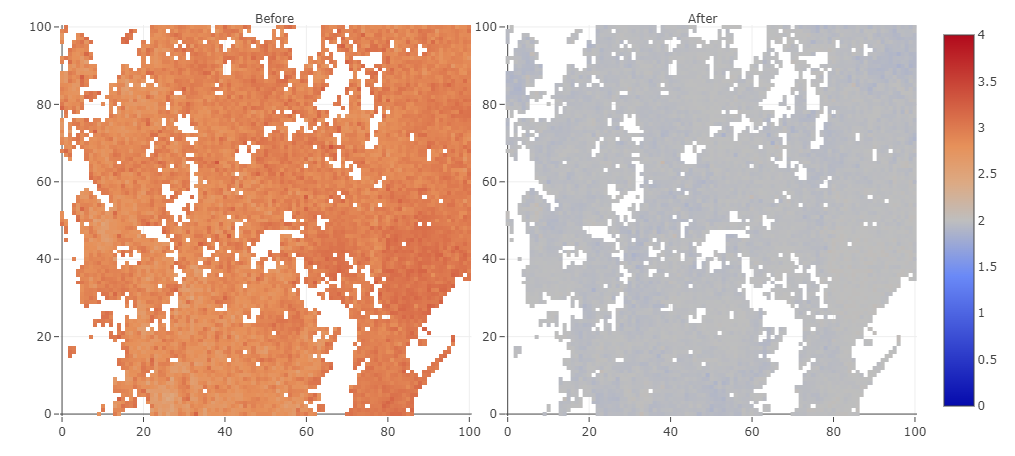}
    \caption{The Durbin-Watson test statistic is reported for each pixel considering the initial temporal series (on the left) and the residuals after fitting the MA(2) model (on the right).}
    \label{fig:dw-test}
\end{figure}

Stationarity is verified using the Augmented Dickey–Fuller test applied to each pixel’s time series. The resulting p-values are uniformly small, leading to rejection of the null hypothesis of non-stationarity.

To determine an appropriate ARMA specification for the series in the second time interval, autocorrelation (ACF) and partial autocorrelation (PACF) functions are examined following the Box–Jenkins method \citep{commandeur2007}. 
For a representative pixel -- whose ACF/PACF behaviour is shared by the other pixels -- the ACF displays a rapid decay and the PACF cuts off after lag 2 (Figure \ref{fig:acf-pacf}). This pattern indicates that an MA(2) specification is appropriate for the entire set of time series.
Model comparison via the Akaike Information Criterion between MA(2), AR(2), MA(3), and AR(3) confirms that MA(2) provides the best fit. Figure \ref{fig:ts-ma2} illustrates the agreement between the observed series and the fitted MA(2) process. To verify that the selected model effectively removes temporal dependence, the Durbin–Watson test is recomputed on the MA(2) residuals. The resulting statistics are uniformly close to 2 across pixels, indicating that autocorrelation has been mitigated.

\begin{figure}[ht]
    \centering
    \includegraphics[width=0.8\linewidth]{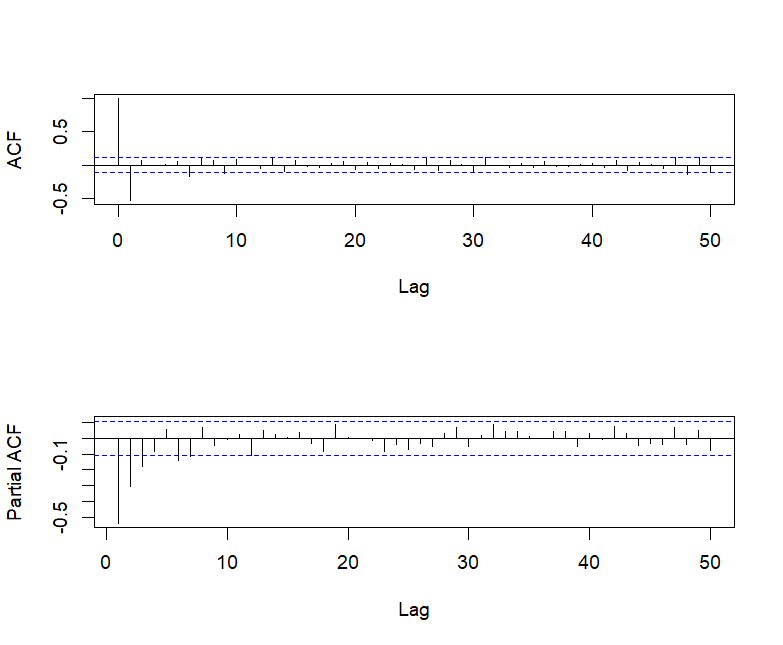}
    \caption{Autocorrelation and partial autocorrelation functions for pixel (50,1) of the red subwindow, considering time instants from the second range of dates.}
    \label{fig:acf-pacf}
\end{figure}

\begin{figure}[ht]
    \centering
    \includegraphics[width=1\linewidth]{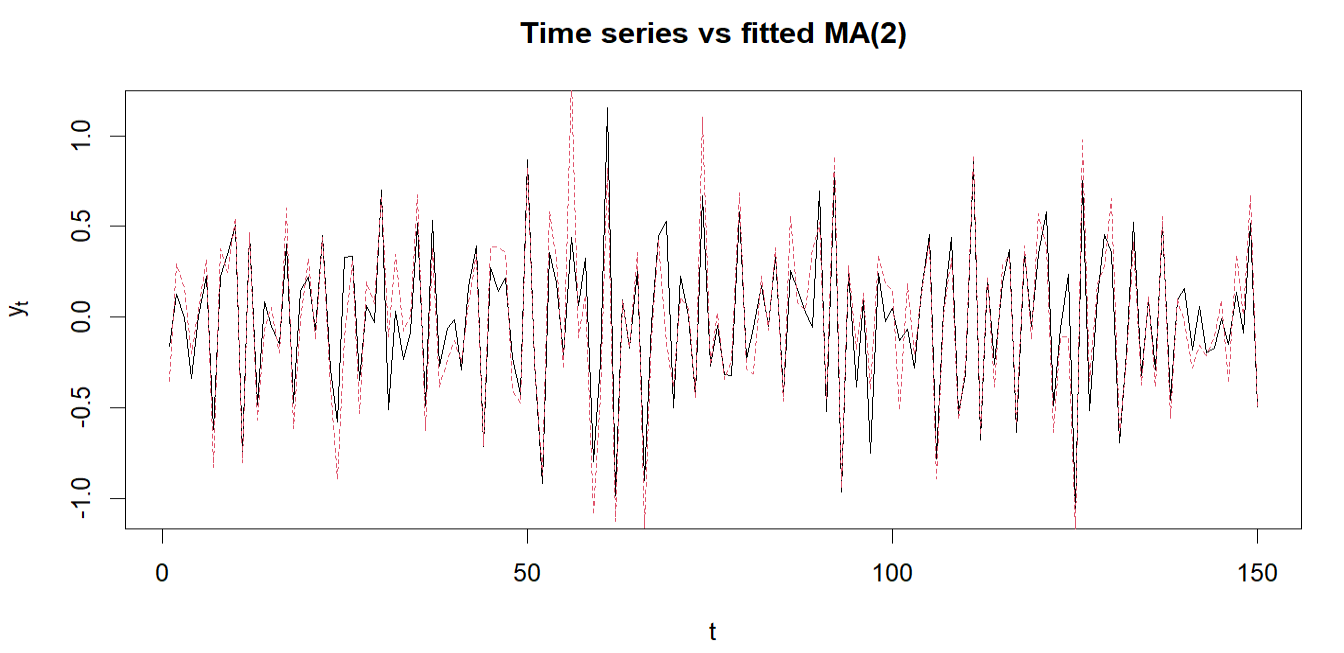}
    \caption{We report the curve of the temporal series of ground displacement for pixel (50,1) in black and its estimation through MA(2) model fitting in a dashed red line, evaluated over a subinterval of time instants of the second range of dates.}
    \label{fig:ts-ma2}
\end{figure}

The same procedure is applied to the remaining two time intervals, associated with 12-day spacing. These shorter series also exhibit negative autocorrelation. Examination of ACF and PACF suggests that AR(1) provides a slightly better fit than MA(1), though the difference is marginal. Given the limited length of these series and to maintain a coherent modeling strategy across time intervals, an MA(1) model is adopted for the 12-day intervals, while the MA(2) model is retained for the 6-day interval.

\subsection{Selection of the hyperparameters}
\label{appendix:case-study-hyp-sel}

\begin{figure}[ht]
    \centering
    \includegraphics[width=0.45\linewidth]{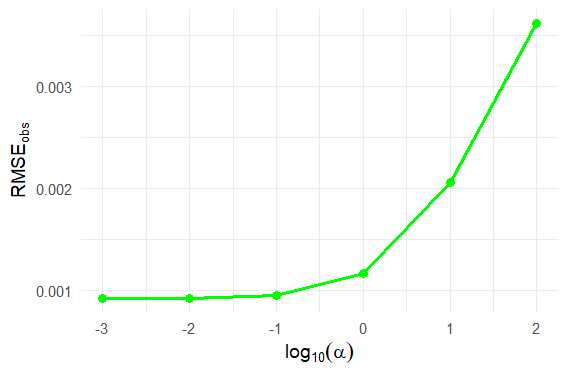}
    \hfill
    \includegraphics[width=0.45\linewidth]{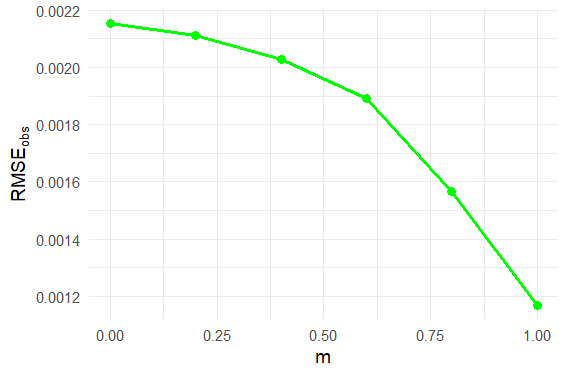}
    \caption{Hyperparameter selection based on RMSE\(_O\).}
    \label{fig:error-alpha-m-cs}
\end{figure}



\end{document}